\documentclass[11pt]{article}
\usepackage{geometry}
\RequirePackage{amsthm,amsmath,amsfonts,amssymb}
\RequirePackage[authoryear]{natbib}
\RequirePackage[colorlinks,citecolor=blue,urlcolor=blue]{hyperref}
\RequirePackage{graphicx}

\RequirePackage{bm}
\RequirePackage{multirow}
\RequirePackage{lscape}
\RequirePackage{algorithm}
\RequirePackage{algcompatible} 

\numberwithin{equation}{section}

\newtheorem{theorem}{Theorem}[section]
\newtheorem{lemma}[theorem]{Lemma}
\newtheorem{corollary}[theorem]{Corollary}
\theoremstyle{remark}

\newtheorem{assumption}{Assumption}
\renewcommand{\hat}{\widehat}
\newcommand{\argmin}{\arg\min}
\newcommand{\mR}{\mathbb{R}}

\newcommand{\E}{{E}}
\newcommand{\PP}{{P}}

\newcommand{\OP}{\mathcal{O}_{P}}
\newcommand{\oP}{o_{P}}

\newcommand{\dell}{\dot{\ell}}
\newcommand{\ddell}{\ddot{\ell}}
\newcommand{\hbeta}{\widehat{\beta}}
\newcommand{\hTheta}{\widehat{\Theta}}
\newcommand{\hSigma}{\widehat{\Sigma} }

\newcommand{\Thetabeta}{\Theta_{\beta^0}}
\newcommand{\hmu}{\widehat{\mu}}
\newcommand{\heta}{\widehat{\eta}}

\newcommand{\lu}[1]{\textcolor{black}{#1}} %

\title{\bfseries \Large De-biased lasso for stratified Cox models with application to the national kidney transplant data}
\author{\large Lu Xia$^a$, Bin Nan$^b$ and Yi Li$^c$}
\date{\small $^a$ Department of Biostatistics, University of Washington, xialu@uw.edu \\
	$^b$ Department of Statistics, University of California, Irvine, nanb@uci.edu \\
	$^c$ Department of Biostatistics, University of Michigan, yili@umich.edu}

\begin{document}


\maketitle

\begin{abstract}
		The Scientific Registry of Transplant Recipients (SRTR) system has become a rich resource for understanding the complex mechanisms of graft failure after kidney transplant, a crucial step for allocating organs effectively and implementing appropriate care. 
		As transplant centers that treated patients might strongly confound graft failures, Cox models stratified by centers can eliminate their confounding effects. Also, 
		since recipient age is a proven non-modifiable risk factor, a common practice is to fit models separately by recipient age groups. The  moderate sample sizes, relative to the number of covariates, in some age groups may lead to biased maximum stratified partial likelihood estimates and unreliable confidence intervals even when samples still outnumber covariates.
		 To draw reliable inference on a comprehensive list of risk factors measured from both donors and recipients in SRTR, we propose a de-biased lasso approach via quadratic programming 
		 for fitting stratified Cox models. We  establish asymptotic properties and verify via simulations that our method produces consistent estimates and confidence intervals with nominal coverage probabilities.
		Accounting for nearly 100 confounders in SRTR, the de-biased method detects that the graft failure hazard nonlinearly increases with donor's age among all recipient age groups, and that organs from older donors more adversely impact the younger recipients. 
		Our method also delineates the associations between graft failure and many risk factors such as recipients' primary  diagnoses  (e.g. polycystic disease, glomerular disease, and diabetes) and donor-recipient mismatches for human leukocyte antigen loci across recipient age groups. These results may inform the refinement of donor-recipient matching criteria for stakeholders.\\[0.5cm]
		\textbf{Keywords:} Confidence intervals; diverging number of covariates; end-stage renal disease; graft failure free survival; statistical inference.
\end{abstract}



\section{Introduction}
	\label{sec:intro}

	
	For patients with end-stage renal disease, one of the most lethal and
	prevalent diseases in the U.S. \citep{saran2020us}, successful renal transplantation is  effective for  improving  quality of life and prolonging survival   \citep{wolfe1999comparison,kostro2016quality,ju2019patient}. The success of kidney transplantation hinges upon  various factors related to the quality of transplant operations, the quality of donated kidneys, and the physical conditions of recipients \citep{rodger2012approach,legendre2014factors}, and  it is crucial to evaluate and understand how these risk factors impact on renal graft failure in order to increase the chance of success \citep{hamidi2016identifying,legendre2014factors}. With the scarcity of organs and an increasing number of waitlisted candidates \citep{bastani2015worsening}, the results can inform more efficient strategies for kidney allocation \citep{rao2009alphabet,smith2012kidney} as well as  evidence-based post-transplant care \citep{baker2017renal}.
	Therefore, how to quantify the impacts of important factors associated with prognosis,  particularly renal graft failure, 
	remains to be a central question in kidney transplantation.
	The Scientific Registry of Transplant Recipients (SRTR) system, a federally funded organization that keeps records of transplant information from  recipients and donors, 
 has become a rich resource for studying post-kidney transplantation prognosis \citep{dickinson2008srtr}.

	Leveraging the SRTR data, one can develop a valid tool for  characterizing the influences of risk factors on graft failure,  a  key step towards post-transplant prognosis. Most previous studies, which focused only on a small number of factors, i.e. kidney diagnosis, recipient age, recipient race, recipient gender, number of human leukocyte antigen (HLA) mismatches,  donor age, donor race, donor gender, serum creatinine level and cold ischemia time \citep{alexander1994effect}, might  have pre-excluded other important factors and not fully captured the complex mechanisms governing graft failure. The SRTR data contain comprehensive information on recipients and donors, such as  recipient primary insurance and employment, procedure type, infection of multiple viruses, history of transplant, transfusion and drug abuse, and pre-transplant comorbidities. The data  provide a unique opportunity for assessing the associations between graft failure and an extended list of variables simultaneously, which may reduce confounding \citep{wang2011gee}. Specifically, since donor age 
	is a major criterion for donor-recipient matching \citep{kasiske2002matching,rao2009comprehensive,veroux2012age}, the data enable us to examine its effect on graft failure by adjusting for confounders, including pre-existing comorbidities.

	There are several statistical challenges. 
	On the one hand, as recipients received care in various transplant centers, the center-specific effects may confound the covariate effects of interest. 
	This motivates us
	to consider Cox models stratified by transplant centers, a commonly used model in the relevant context \lu{without the need to explicitly model the potentially time-varying center effects  \citep{he2021stratified}}.
	On the other hand, 
	recipient age is a strong risk factor and there may exist
	complicated interactions between recipients' age and other characteristics \citep{keith2004effect}.
	For ease of interpretation and by convention \citep{morales2012risk,faravardeh2013predictors},   we have opted to divide our analyzable patient population (the adult recipients with 
	kidneys transplanted during 2000 and 2001) into  $[18, 45]$, $(45, 60]$ and $60+$ years old groups (Table \ref{tab:character}), and
	fit models separately for these three groups. Allowing model parameters to be age group-specific, we have avoided parametrically modeling the interactions between recipient age and the other risk factors.
	When the number of covariates is relatively large 
	(94 in our data)  compared to, though still less than, the sample size
	(for example, 1448 patients with 1013 events in the $60+$ years old recipient group), the conventional maximum stratified partial likelihood estimation (MSPLE) may yield untrustworthy point estimates,   confidence intervals and hypothesis testing results, as illustrated in our simulations.

		\begin{table}[ht]
		\centering
		\caption{Study population characteristics by recipient age group} 
		\begin{tabular}{llll}
			\hline
			Recipient age group & $[18, 45]$ & $(45, 60]$ & $60+$ \\
			Variable &  \multicolumn{3}{c}{Mean (SD) / Count (\%) }  \\ \hline
			\# Centers  & 84 (--) & 107 (--)  & 43 (--)  \\
			\# Patients & 3388 (100\%) & 4359 (100\%) & 1448 (100\%) \\
			\# Events & 1588 (46.9\%) & 2334 (53.5\%)  & 1013 (70.0\%)  \\
			Recipient age & 35.7 (7.0) & 53.0 (4.2)  &  66.6 (4.3)  \\
			Donor age (years) &   &   &  \\
			$\qquad$ $\le 10$        & 276 (8.1\%) & 223 (5.1\%) & 61 (4.2\%) \\
			$\qquad$ $(10, 20]$    & 580 (17.1\%)  & 611 (14.0\%)  & 137 (9.5\%) \\	
			$\qquad$ $(20, 30]$    & 633 (18.7\%) & 683 (15.7\%) & 179 (12.4\%) \\
			$\qquad$ $(30, 40]$    & 505 (14.9\%)  & 599 (13.7\%)  & 174 (12.0\%) \\
			$\qquad$ $(40, 50]$    & 753 (22.2\%) & 947 (21.7\%)  & 256 (17.7\%) \\
			$\qquad$ $(50, 60]$    & 498 (14.7\%) & 893 (20.5\%) & 318 (22.0\%) \\
			$\qquad$ $60+$    & 143 (4.2\%) & 403 (9.2\%) & 323 (22.3\%) \\
			Recipient gender &   &   &  \\
			$\qquad$ Male        & 1997 (58.9\%) & 2671 (61.3\%)  & 913 (63.1\%)  \\
			$\qquad$ Female        & 1391 (41.1\%) & 1688 (38.7\%)  & 535 (36.9\%)   \\
			Donor gender &   &   &  \\
			$\qquad$ Male        & 2039 (60.2\%) & 2563 (58.8\%)  & 803 (55.5\%)  \\
			$\qquad$ Female        & 1349 (39.8\%) & 1796 (41.2\%) & 645 (44.5\%)  \\
			\hline
		\end{tabular}
		\label{tab:character}
	\end{table}


	For proper inferences, we consider an asymptotic framework with a diverging number of covariates, wherein the number of covariates, though smaller than the sample size, can increase with the sample size \citep{he2000parameters,wang2011gee}. 
	 \lu{Lasso provides a very popular tool for simultaneous variable selection and estimation with high-dimensional covariates \citep{tibshirani1997lasso}.  For unstratified Cox models, \citet{huang2013oracle} and \citet{kong2014non} presented the oracle inequalities for the lasso estimator.  However,  with penalization,  lasso estimates are biased towards zero \citep{van2014asymptotically},  and they do not possess regular limiting distributions even under linear regression with a fixed number of covariates \citep{fu2000asymptotics}. Conditional inference based on the selected model is invalid, either, due to the failure to account for uncertainty in model selection.  Hence,  lasso cannot be directly applied to draw statistical inference.}
	There is  literature \lu{on inference for} unstratified Cox proportional hazards models
	under the related asymptotic framework.
	For example, \citet{fang2017testing} proposed decorrelated score tests for a low-dimensional component in the regression parameters, and \citet{kong2018high}, \citet{yu2018confidence} and \citet{xia2021cox} proposed to correct asymptotic biases of the lasso estimator following the framework of  \citet{van2014asymptotically}, \citet{zhang2014confidence} and \citet{javanmard2014confidence} that were originated from high-dimensional linear or generalized linear models. For Cox models, all of these methods, except \citet{xia2021cox} which considered the ``large $n$, diverging $p$'' scenario, assumed sparsity on the inverse information matrix. 
	This sparse matrix assumption, however, may not hold for models beyond linear regression, leading to insufficient bias correction and  under-covered  confidence intervals. Moreover, as  these methods were not designed for modeling stratified data, they are not directly applicable to the analysis of the SRTR data. \lu{To our knowledge, the current literature lacks inferential methods with  theoretical rigor for stratified Cox models with a diverging number of covariates.}

	We propose a 
	de-biased lasso approach  for Cox models 
	stratified by transplant centers,
	which  solves a series of quadratic programming problems to estimate the inverse information matrix, and corrects the biases from the lasso estimator for valid statistical inference. Our asymptotic results enable us to draw inference   on any linear combinations of model parameters,  including the low-dimensional targets in \citet{fang2017testing} and  \citet{kong2018high} as special cases and  fundamentally deviating  from the stepwise regression adopted by \citet{rao2009comprehensive}. 
	\lu{When  the number of covariates is relatively large compared to the sample size, 
	our approach yields less biased estimates and more properly covered confidence intervals than MSPLE as well as the methods of \citet{fang2017testing, kong2018high, yu2018confidence}  adapted to the stratified setting.} Therefore, it is well-suited  
	for analyzing the SRTR data, especially among the oldest recipient group that has the smallest sample size.

	Applications of our method to the SRTR data have generated   reliable estimation and inference results for the effects of an expanded list of donor and recipient factors. 
	We find that receiving kidneys from older donors is associated with an increased hazard of graft failure after adjusting for many confounding factors, and that the dependence on donors' age is non-linear.
	The results may inform more comprehensive assessments of  post-transplant prognosis and  kidney allocation.

	The article is organized as follows. We introduce the proposed de-biased lasso approach in Section \ref{sec:method} and establish the asymptotic results in Section \ref{sec:theory}, which form the basis of inference for the SRTR data.  We conduct simulations  in Section \ref{sec:simul} and demonstrate that our  method outperforms MSPLE in  bias correction and confidence interval coverage. In Section \ref{sec:app}, we analyze the SRTR data by using the proposed de-biased approach.
    Finally, we provide a few concluding remarks in Section \ref{sec:discuss}, the detailed list of covariates considered in the analysis of SRTR data in Appendix \ref{appA}, and regularity conditions in Appendix \ref{appB}. Technical details and proofs are deferred to the Supplementary Material.


	\section{De-biased lasso for stratified Cox models via quadratic programming}
	\label{sec:method}

	We apply stratified Cox models to evaluate the impacts of risk factors on post-transplant graft failure using the SRTR data.  For each  recipient age group defined in the first row of Table \ref{tab:character}, let $K$ be the total number of transplant centers, and $n_k$ be the number of recipients in the $k$-th transplant center, $k=1, \cdots, K$. With $i$ indexing recipients within the $k$-th transplant center, let $T_{ki}$ denote the graft failure free survival time, i.e. the time from transplantation to graft failure or death, whichever comes first [a common endpoint in transplantation \citep{kasiske2011relationship}],  $X_{ki}$ be  a $p$-dimensional covariate vector,
	and $C_{ki}$ be the censoring time. We assume random censoring, that is,  
	$T_{ki}$ and $C_{ki}$ are independent given $X_{ki}$. In the SRTR data, $p=94$ and  $X_{ki}$ includes risk factors from both donors and recipients, such as  gender, ABO blood type, history of diabetes and duration, angina/coronary artery disease, symptomatic peripheral vascular disease, drug treated systemic hypertension, drug treated COPD, and mismatch for each HLA locus between donors and recipients; see	a full list of covariates in Appendix \ref{appA}.
	 Let $\delta_{ki} = 1(T_{ki} \le C_{ki})$ be the event indicator and  $Y_{ki} = \min(T_{ki}, C_{ki})$ be the observed time. 
	With a center-specific baseline hazard function  $\lambda_{0k}(t)$, a stratified Cox model for $T_{ki}$ stipulates that its conditional hazard at $t$ given $X_{ki}$ is 
	\[
	\lambda_{ki}(t | X_{ki}) = \lambda_{0k}(t) \exp\{X_{ki}^T \beta^0\},
	\]
	where $\beta^0 = (\beta^0_1, \ldots, \beta^0_p)^T \in \mR^p$ is the vector of common regression coefficients across all centers. \lu{It is reasonable to assume that the true regression coefficients
$\beta^0$ are the same across strata \citep{kalbfleisch2002statistical}, while the center effects, though not of primary interest here, are accounted for via different baseline hazards $\lambda_{0k}(t)$'s.}
	
	\subsection{Estimation method}
	\label{subsec:dslasso}

	The MSPLE of  $\beta$  minimizes the following {\em negative} log stratified partial likelihood function
	\begin{equation} \label{eq:negloglik1}
	\ell(\beta) =  -\displaystyle \frac{1}{N} \sum_{k=1}^K \sum_{i=1}^{n_k} \left[  \beta^T X_{ki} - \log \left\{ \frac{1}{n_k}  \sum_{j=1}^{n_k} 1(Y_{kj} \ge Y_{ki}) \exp(\beta^T X_{kj}) \right\}   \right] \delta_{ki},
	\end{equation}
	where $N = \sum_{k=1}^K n_k$.
	In SRTR, the number of risk factors, though smaller than the sample size, is fairly large. In this case,  
	our numerical examination shows that  MSPLEs are biased and their confidence intervals do not yield nominal coverage. We consider a de-biased approach that has been shown to yield valid inference in linear regression \citep{van2014asymptotically,zhang2014confidence,javanmard2014confidence}.
Here we assume that $p<N$ but grows with $N$, which falls into the ``large $N$, diverging $p$'' framework. We extend  the de-biased lasso  to accommodate stratified Cox models.

	
	For a vector $x = (x_1, \ldots, x_p)^T \in \mR^p$, define $x^{\otimes 0} = 1$, $x^{\otimes 1} = x$ and $x^{\otimes 2} = x x^T$.
	Let $\dot{\ell}(\beta)$ and $\ddot{\ell}(\beta)$ be the first and the second order derivatives of $\ell(\beta)$ with respect to $\beta$, i.e.
	\begin{align*}
	\dell(\beta)  &=  - \displaystyle \frac{1}{N} \sum_{k=1}^K   \sum_{i=1}^{n_k} \left\{ X_{ki} - \displaystyle \frac{\hmu_{1k}(Y_{ki}; \beta) }{\hmu_{0k} (Y_{ki}; \beta) }  \right\} \delta_{ki}, \\
	\ddell(\beta) & = \displaystyle  \frac{1}{N} \sum_{k=1}^K \sum_{i=1}^{n_k} \left\{  \displaystyle \frac{\hmu_{2k}(Y_{ki}; \beta)}{\hmu_{0k}(Y_{ki}; \beta)} - \left[ \displaystyle \frac{\hmu_{1k}(Y_{ki}; \beta)}{\hmu_{0k}(Y_{ki}; \beta)} \right]^{\otimes 2} \right\} \delta_{ki},
	\end{align*}
	where $\hmu_{rk} (t; \beta) = {n_k}^{-1} \sum_{j=1}^{n_k} 1(Y_{kj} \ge t) X_{kj}^{\otimes r} \exp\{X_{kj}^T \beta\}, ~ r=0, 1, 2$. 
The lasso estimate, $\hbeta$, minimizes the penalized negative log stratified partial likelihood, 
	\begin{equation} \label{eq:lasso}
	\hbeta = {\argmin}_{\beta \in \mR^p} \{ \ell(\beta) + \lambda \| \beta \|_1 \},
	\end{equation}
	where $\lambda > 0$ is a tuning parameter that encourages sparse solutions. Here,  $\| x \|_q = (\sum_{j=1}^p |x_j|^q)^{1/q}$ is the $\ell_q$-norm for  $x \in \mR^{p}$, $q \ge 1$. 
	
	As $\hbeta$ is typically biased,  we can obtain the de-biased lasso estimator by a Taylor expansion of $\dot{\ell}(\beta^0)$ around $\widehat{\beta}$. To proceed,
	let  $\widehat{M}$ be a $p \times p$ matrix and $\widehat{M}_j$ its $j$th row. Pre-multiplying $\widehat{M}_j$ on both sides of the Taylor expansion and collecting terms, we have the following equality for the $j$th component of $\beta$:
	\begin{equation} \label{eq:derive_bhat_M}
	\widehat{\beta}_j  -  \beta^0_j  + \overbrace{ \left( - \widehat{M}_j \dot{\ell}(\widehat{\beta})   \right) }^{I_j}
	+ \overbrace{\left( - \widehat{M}_j {\Delta} \right) }^{II_j} 
	+ \overbrace{\left( \widehat{M}_j \ddot{\ell}(\widehat{\beta}) - {e}_j^T \right)  \left( \widehat{\beta} - \beta^0 \right)}^{III_j}   = - \widehat{M}_j \dot{\ell}(\beta^0),
	\end{equation}
	where the remainder $\Delta \in \mR^p$ in $II_j$ can be shown asymptotically negligible given the convergence rate of the lasso estimator $\widehat{\beta}$, and so is $III_j$ if $\widehat{M}_j \ddot{\ell}(\widehat{\beta}) - e_j^T$ converges to zero with certain rate that will be discussed later in Section \ref{sec:theory}. 
	Hence, the de-biased lasso estimator 
	corrects the bias of  $\hbeta_j$ with a one-step update of
	\begin{equation} \label{eq:dslasso0}
	\widehat{b}_j = \hbeta_j - \hTheta_j \dot{\ell}(\hbeta),
	\end{equation}
	which replaces $\widehat{M}_j$ in \eqref{eq:derive_bhat_M} with the $j$-th row of $\hTheta$, an estimate of the inverse information matrix $\Thetabeta$, and $ - \hTheta_j \dot{\ell}(\hbeta)$ is the bias correction term
	to $\hbeta_j$. Here,  $\Thetabeta$  is the inverse of the population version of $\hSigma$ given in the following  \eqref{eq:hatsigma}; see the explicit definition of $\Thetabeta$ underneath \eqref{popsig}.
		  Denote by
	$\widehat{b} = (\widehat{b}_1, \ldots, \widehat{b}_p)^T$ the vector of the de-biased lasso estimates, and, for compactness,  write (\ref{eq:dslasso0}) in a matrix form
	\begin{equation} \label{eq:dslasso}
	\widehat{b} = \hbeta - \hTheta \dot{\ell}(\hbeta).
	\end{equation}
	 Unlike $\hbeta$, the  de-biased estimator $\widehat{b}$ in (\ref{eq:dslasso}) is no longer  sparse. Motivated by \citet{javanmard2014confidence} on high-dimensional inference in linear regression, we propose to obtain $\hTheta$ by solving a series of quadratic programming problems. First, we compute
	\begin{equation} \label{eq:hatsigma}
	\hSigma = 
	\displaystyle \frac{1}{N}\sum_{k=1}^K \sum_{i=1}^{n_k} \delta_{ki} \left[  X_{ki} - \widehat{\eta}_k(Y_{ki}; \hbeta)  \right]^{\otimes 2},
	\end{equation}
	where $\widehat{\eta}_k(t; \beta) = \hmu_{1k}(t; \beta) / \hmu_{0k}(t; \beta)$ is the vector of weighted average covariates. 
	We use $\widehat{\Sigma}$, in lieu of $\ddot{\ell} (\hbeta)$,  for ease of proving  theoretical properties. 
	Indeed, as shown in the Supplementary Material,  $\| \hSigma - \ddot{\ell} (\hbeta) \|_{\infty} \stackrel{p}{\rightarrow}0 $ with a desirable  rate under the conditions in Section \ref{sec:theory}.  Next,  for each $j = 1, \ldots, p$, 
	we solve  a quadratic programming problem
	\begin{equation} \label{eq:qp}
	\min_{ m \in \mR^p} \left\{  m^T \hSigma m:  \| \hSigma m - e_j \|_{\infty} \le \gamma  \right\},
	\end{equation}
	where $\gamma > 0$ is a tuning parameter that is different from the lasso tuning parameter $\lambda$, $e_j$ is a unit directional vector with only the $j$th element being one, and $\| \cdot \|_{\infty}$ is the matrix max norm, i.e. $\| A \|_{\infty} = \max_{i,j} |A_{ij}|$ for a real matrix $A $. 
	Denote by $m^{(j)}$ the column vector of solution to (\ref{eq:qp}). We
	 obtain a $p \times p$ matrix $\hTheta = (m^{(1)}, \ldots, m^{(p)})^T$.  
	
	The constraint $\| \hSigma m - e_j \|_{\infty} \le \gamma$ in \eqref{eq:qp} controls deviations of  the de-biased estimates from the lasso estimates. In an extreme case of $\gamma = 1$, an admissible solution is $m=0$, and therefore there is no bias correction in the de-biased estimator;
	in another extreme case of $\gamma=0$, $m^{(j)}$ is the $j$th column of $\hSigma^{-1}$.  
	We  implement  \eqref{eq:qp}  by using  R  \texttt{solve.QP()}, which  can be programmed in parallel for large $p$. We name the method {\sl de-biased lasso via quadratic programming} (hereafter, DBL-QP).

	\subsection{Tuning parameter selection}
	\label{subsec:tuning}
	
	For the DBL-QP method, the lasso tuning parameter $\lambda$ can be selected via 5-fold cross-validation 
	as in \citet{simon2011regularization}.
	The selection of $\gamma$ is crucial
	as, for example,  Figure \ref{fig:tuning} reveals that $\gamma$ should be selected within a specific  range  (shaded in figures) to achieve the most desirable bias correction and confidence interval coverage probability. It also shows the large bias and poor coverage resulting from MSPLE. Inappropriate tuning can yield even more biased estimates with poorer coverage than  MSPLE. Results of lasso and oracle estimates are also provided as references, where oracle estimates are obtained from the reduced model that only contains truly nonzero coefficients.  
	
	\begin{figure}
		\centering
		\includegraphics[width=0.6\textwidth]{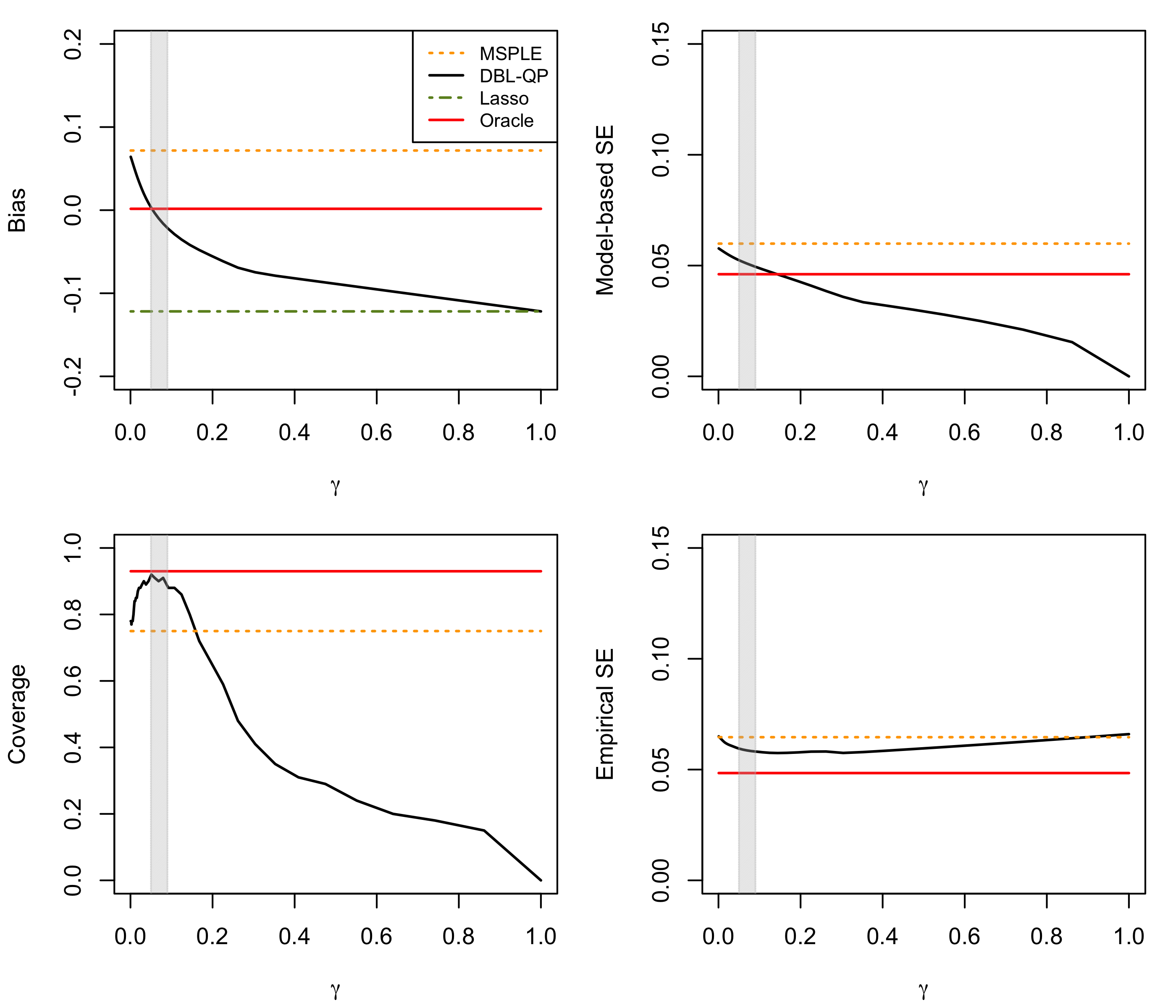}
		\caption{ The impact of choices of $\gamma$ on the averages of biases, empirical coverage probabilities, model-based and empirical standard errors, based on 100  simulations with $K=5$ strata and $n_k=200$  in each stratum,   $p=100$ covariates simulated from a multivariate normal distribution with mean zero and an AR(1) covariance matrix  ($\rho=0.5$) and truncated at $\pm 3$. 
		 Survival times are simulated with a hazard $\lambda_{0k} \exp\{X_{ki}^T \beta^0\}$, where $\lambda_{0k}$ are  constants generated from Uniform $(0.1, 0.5)$, and four nonzero coefficients in $\beta^0$ take  1, 1, 0.3 and 0.3, respectively.
		 Censoring times are independently simulated from  Uniform $(1, 30)$. 
	}
		\label{fig:tuning}
	\end{figure}


	Intuitively, $\gamma$ should be chosen near zero, resulting in a de-biased estimator with estimates of large coefficients close to oracle estimates. We do not recommend evaluating  cross-validation criteria by  plugging in the de-biased estimates because of  accumulative estimation errors. We opt for a  hard-thresholding approach that more effectively removes noise from the de-biased estimates: we retain the de-biased lasso estimate for $\beta^0_j$  only if the null hypothesis $\beta_j^0 = 0$ is rejected; otherwise, we set it  to zero  (shown in Algorithm  \ref{algo1}). The set $\widehat{A}$ in Step 2.2 of Algorithm  \ref{algo1} is expected to estimate well the set of truly associated variables. Specifically, we set $\widehat{A}$ to be the index set of variables whose Wald  statistic $ \sqrt{N} | \widehat{b}_j | /\widehat{\Theta}_{jj}^{1/2} > z_{\alpha/(2p)}$, where $z_{\alpha/(2p)}$ is the upper $\{ \alpha/(2p) \}$-th quantile of the standard normal distribution. The cutoff is determined by Theorem \ref{thm:main} and Bonferroni correction for multiple testing.
	When implementing cross-validation,  we can either take stratum as the sampling unit and randomly split  strata, or randomly split observations within each stratum, to form training and testing subsets. 
  We find  the former 
improves stability of tuning parameter selection when there are a number of small-sized strata. 

	\begin{algorithm}[ht] 	
		\caption{Selection of the tuning parameter $\gamma$ using cross-validation \label{algo1}}
		\begin{algorithmic}
			\item[\textbf{Step 1}]  Pre-determine a grid of points for $\gamma$ in [0,1], denoted as $\gamma^{(g)}, g=1, \cdots, G$, and set each $cv_g = 0$.  
			\item[\textbf{Step 2}] Randomly assign the $K$ strata into $M$ folds, 
			leaving one fold for testing and the others for training. 
			Set $q = 1$.  
			\item[\quad \textbf{Step 2.1}]  While $q \le M$, use the $q$th training set to compute the de-biased lasso estimator with $\gamma^{(g)}, g=1, \cdots, G$, denoted as $\widehat{b}^{(gq)}$, and define the active set $\widehat{A}^{(gq)}$. 
			\item[\quad \textbf{Step 2.2}] Define the thresholded de-biased lasso estimator $\widehat{b}^{(gq)}_{thres} = \widehat{b}^{(gq)} \cdot 1(j \in \widehat{A}^{(gq)})$, i.e. setting components of $\widehat{b}^{(gq)}$outside the active set $\widehat{A}^{(gq)}$ to 0.  
			\item[\quad \textbf{Step 2.3}]  Compute the negative log partial likelihood on the $q$th testing set $\ell^{(q)} (\widehat{b}^{(gq)}_{thres})$. 
			\item[\quad \textbf{Step 2.4}]  Set $cv_g \leftarrow cv_g + N^{(q)} \ell^{(q)} (\widehat{b}^{(gq)}_{thres})$, for $g = 1, \cdots, G$, where $N^{(q)}$ is the total number of observations in the $q$th testing set.  
			\item[\quad \textbf{Step 2.5}] Set $q \leftarrow q + 1$ and go to Step 2.1.
			\item[\textbf{Step 3}] Let $\hat{g} = \argmin_{g} cv_g$. The final output tuning parameter value is $\gamma^{(\hat{g})}$.
		\end{algorithmic}  
	\end{algorithm}

\section{Valid statistical inference based on the de-biased lasso estimator}
	\label{sec:theory}
	
	This section  presents asymptotic results, which lay the groundwork for using the de-biased lasso estimator described in Section \ref{sec:method} to infer on the risk factors of graft failure in the SRTR analysis. The pertaining large sample framework posits that the number of strata $K$ is  fixed,  the smallest stratum size $n_{min} = \min_{1 \le k \le K} n_k \rightarrow \infty$, and 
	$
	{n_k} / {N} \rightarrow r_k > 0
	$
	as $n_{min} \rightarrow \infty$, $k=1, \cdots, K$.  
 \lu{This framework  conforms to the real world setting of our concern, where the number of transplant centers nationwide is finite, and the number of patients or transplant events in each center increases over the years.}
	We provide regularity conditions \lu{and their discussion} in Appendix \ref{appB},  \lu{and present all} the proofs in the Supplementary Material.

	Let $\mu_{rk} (t; \beta) = \E [1(Y_{k1} \ge t) X_{k1}^{\otimes r} \exp\{X_{k1}^T \beta\}]$  be the limit of $\hmu_{rk} (t; \beta)$, $r = 0, 1, 2, ~ k=1, \cdots, K$. Then the limit of the weighted covariate process for $\heta_k(t; \beta) = \hmu_{1k} (t; \beta) / \hmu_{0k} (t; \beta)$ becomes $\eta_{k0}(t; \beta) = \mu_{1k} (t; \beta) / \mu_{0k} (t; \beta)$. Let
	\[
	\Sigma_{\beta^0, k} = \E [\{ X_{ki} - \eta_{k0}(Y_{ki}; \beta^0) \}^{\otimes 2} \delta_{ki}]
	\] 
	be the information matrix for the $k$-th stratum, $k=1, \cdots, K$. The overall information matrix across all strata then becomes the weighted average of the stratum-specific information matrices,
	\begin{equation} \label{popsig}
	\Sigma_{\beta^0} = \sum_{k=1}^K r_k \Sigma_{\beta^0, k}.
	\end{equation}
	The inverse information matrix is $\Theta_{\beta^0} = \Sigma_{\beta^0}^{-1}$, which is to be approximated by  $\hTheta$ obtained in Section \ref{subsec:dslasso}. 
	
	The following theorem  establishes the asymptotic normality of any linear combination of the estimated  regression parameters, {$c^T \widehat{b}$} for some loading vector $c \in \mR^p$, obtained by the proposed DBL-QP method. For an $m \times r$ matrix $A = (a_{ij})$,  define the $\ell_1$-induced matrix norm $\| A \|_{1,1} = \max_{1 \le j \le r} \sum_{i=1}^m |a_{ij}|$. For two positive sequences $\{a_n\}$ and $\{b_n\}$, we write $a_n \asymp b_n$ if there exist two constants $C$ and $C^{\prime}$ such that $0 < C \le a_n / b_n \le C^{\prime} < \infty$. Let $s_0$ be the number of nonzero elements of $\beta^0$.

	\begin{theorem} \label{thm:main}
		Assume that the tuning parameters $\lambda$ and $\gamma$ satisfy $\lambda \asymp \sqrt{\log(p)/n_{min}}$ and \\ $\gamma \asymp \| \Theta_{\beta^0} \|_{1,1} \{ \max_{1\le k \le K} | n_k/N - r_k| + s_0 \lambda \} $, and that $\| \Theta_{\beta^0} \|_{1,1}^2  \{ \max_{1 \le k \le K}|n_k/N - r_k| + s_0 \lambda \} p\sqrt{\log(p)} \rightarrow 0$ as $n_{min} \rightarrow \infty$.
		  Under  Assumptions \ref{assump1}--\ref{assump5} given in  Appendix \ref{appB}, for any $c \in \mathbb{R}^p$ such that $\| c \|_2 = 1$, $\| c \|_1 \le a_*$ with  $a_* < \infty$ being an absolute positive constant, and $\{ c^T \Thetabeta c \}^{-1} = \mathcal{O}(1)$, we have
		\[
		\displaystyle \frac{\sqrt{N} c^T (\widehat {b} - \beta^0)}{(c^T \widehat{\Theta} c)^{1/2} }
		\overset{\mathcal{D}}{\rightarrow} \mathcal{N}(0,1).
		\]
		
	\end{theorem}

\lu{Note that, instead of listing it as a regularity condition in Appendix \ref{appB}, we assume $\{ c^T \Thetabeta c \}^{-1} = \mathcal{O}(1)$ in the above theorem because the vector $c$ is also defined here. A similar condition is assumed in \cite{van2014asymptotically} [Theorem 3.3 (vi)] which is weaker than uniformly bounding the maximum eigenvalue of $\Sigma_{\beta^0}$.	}  The hypothesis testing with $H_{0}: c^T \beta^0 - a_0=0$ versus $H_{1}: c^T \beta^0- a_0 \ne 0$ for some constants $c \in \mR^p$ and $a_0$
	  entails various applications. 
	 For example, by setting $a_0=0$ and $c$ to be a basis vector with only one element being 1 and all the others 0, we can draw
	 inference on any  covariate in the presence of all the other covariates. 
	 In particular, we will draw inference on the pairwise differences in graft failure risk 
	 among donor age groups, e.g. between $(10, 20]$ and $(20,30]$ (the reference level) years old, and among patients with different primary kidney diagnoses (diabetes is the reference level); see Section \ref{sec:app}.
	Given an appropriately chosen $c$ and with $T = \sqrt{N} (c^T \widehat{b} - a_0) / (c^T \widehat{\Theta} c)^{1/2}$,  we  construct a two-sided test function 
	\[
	\phi(T) = \left\{  \begin{array}{ll}
	1 & \quad \mathrm{if} \ |T| > z_{\alpha/2} \\
	0 & \quad \mathrm{if} \ |T| \le z_{\alpha/2}
	\end{array} ,  \right.
	\]
	where $z_{\alpha/2}$ is the upper $(\alpha/2)$-th quantile of the standard normal distribution. Corollary \ref{coro:power} provides the asymptotic type I error and power of the test $\phi(T)$, and  Corollary \ref{coro:ci} formalizes the construction of  level $\alpha$ confidence intervals for $c^T \beta^0$ which ensures  the nominal coverage probability asymptotically.

	\begin{corollary} \label{coro:power}
		Under the conditions specified  in Theorem \ref{thm:main}, 
		$\PP (\phi(T)=1 | H_0) \rightarrow \alpha$ as $n_{min} \rightarrow \infty$. Moreover, under  $H_1: a_0 - c^T \beta^0 \ne 0$,  
		$\PP (\phi(T) = 1 | H_1) \rightarrow 1$. 
	\end{corollary}

	\begin{corollary} \label{coro:ci}
		Suppose that the conditions in Theorem \ref{thm:main} hold. 
		Construct the random confidence interval $\mathcal{R}(\alpha) = \left[ c^T \widehat{b} - z_{\alpha/2} (c^T \widehat{\Theta} c / N)^{1/2},~ c^T \widehat{b} + z_{\alpha/2} (c^T \widehat{\Theta} c / N)^{1/2}\right]$. Then $\PP (c^T \beta^0 \in \mathcal{R}(\alpha)) \rightarrow 1 - \alpha$ as $n_{min} \rightarrow \infty$, where the probability is taken 
		under the true $\beta^0$.
	\end{corollary}
	
	Our asymptotic results facilitate simultaneous inference on multiple contrasts in the context of post-transplant renal graft failure. 
	For example, an important question to address is whether donor age is associated with graft failure. 
	With  categorized  donor age in our data analysis, simultaneous comparisons among the seven categories, e.g. $\le 10, (10, 20],
	(20,30], (30,40], (40,50], (50, 60]$ and $60+$, naturally form multiple null contrasts. 
	 These contrasts can be formulated by $J \beta^0$, where $J$ is an $m\times p$ matrix, and $m$ represents the number of linear combinations or contrasts.
	 The following theorem and corollary summarize the results for inference on multiple contrasts, $J \beta^0$.  See an application of the asymptotic results to the SRTR data with $(m,p)=(6,94)$ in Section \ref{sec:app}. 
	
	\begin{theorem} \label{thm:simul}
		Suppose that $J$ is an  $m\times p$ matrix  with $rank(J) = m$, $\| J \|_{\infty, \infty} = \mathcal{O}(1)$ and $J \Theta_{\beta^0} J^T \rightarrow F$, where $F$ is a nonrandom $m \times m$ positive definite matrix. Assume that the tuning parameters $\lambda$ and $\gamma$ satisfy $\lambda \asymp \sqrt{\log(p)/n_{min}}$ and $\gamma \asymp \| \Theta_{\beta^0} \|_{1,1} \{ \max_{1\le k \le K} | n_k/N - r_k| + s_0 \lambda \} $, and that $\| \Theta_{\beta^0} \|_{1,1}^2  \{ \max_{1 \le k \le K}|n_k/N - r_k| + s_0 \lambda \} p\sqrt{\log(p)} \rightarrow 0$ as $n_{min} \rightarrow \infty$. Under  Assumptions \ref{assump1}--\ref{assump3}, \ref{assump5} and \ref{assump6} given in Appendix \ref{appB}, we have
		\[
		\sqrt{N} J (\widehat{b} - \beta^0) \overset{\mathcal{D}}{\rightarrow} \mathcal{N}_m (0,F).
		\]
	\end{theorem}
	Here, $\| A \|_{\infty,\infty} = \max_{1 \le i \le m} \sum_{j=1}^r |a_{ij}|$ is the $\ell_{\infty}$-induced matrix norm for an $m \times r$ matrix $A = (a_{ij})$. 
	The theorem implies the following corollary, which  constructs test statistics and multi-dimensional confidence regions with  proper asymptotic type I error rates and nominal coverage probabilities. 
	
	\begin{corollary} \label{coro:thm2}
		Suppose the conditions in Theorem \ref{thm:simul} hold. For an $m \times p$ matrix $J$ as specified  in Theorem \ref{thm:simul}, and under  $H_0: J \beta^0 = a^0  \in \mR^m$,   
		\[T^{\prime} = N (J \widehat{b} - a^0)^T \widehat{F}^{-1} (J \widehat{b} - a^0)  \overset{\mathcal{D}}{\rightarrow} \chi^2_m,
		\]
		where $\widehat{F} = J \hTheta J^T$. Moreover, for an $\alpha \in (0,1)$, 
		define the random set $\mathcal{R}^{\prime}(\alpha) = \{ a \in \mR^m: N (J \widehat{b} - a)^T \widehat{F}^{-1} (J \widehat{b} - a) < \chi^2_{m, \alpha} \}$, where $\chi^2_{m, \alpha}$ is the upper $\alpha$-th quantile of $\chi^2_m$. Then $\PP (J\beta^0 \in \mathcal{R}^{\prime}(\alpha)) \rightarrow 1 - \alpha$ as $n_{min} \rightarrow \infty$, where the probability is taken under the true $\beta^0$.
	\end{corollary}

	\section{Simulation study}
	\label{sec:simul}
	
	We conduct simulations to  examine the finite sample performance of the proposed  DBL-QP approach  in correcting  estimation biases and maintaining nominal   coverage probabilities of confidence intervals.  For comparisons, we also perform MSPLE, the oracle estimation,  \lu{and the three inference methods [``Nodewise'' for \citet{kong2018high}, ``CLIME'' for \citet{yu2018confidence},  and ``Decor'' for \citet{fang2017testing}] that are adapted to stratified Cox models}.  The following  scenarios pertain to four combinations of $(K, n_k, p)$, 	where $K, n_k$ and $p$ are the number of strata, stratum-specific sample size and the number of covariates, respectively.  Specifically,  Scenarios  1--3 refer to  $(K, n_k, p)=(10, 100,10)$, $(10, 100,100)$, and $(5, 200,100)$, respectively. 
 In Scenario 4, $K=40$, $p = 100$, $n_k$'s are simulated from a Poisson distribution with mean 40 and then fixed in all of the replications. This scenario mimics the situation of the recipient group aged over 60, the smallest group in the SRTR data.

	Covariates $X_{ki}$ are simulated from $N_p (0, \Sigma_x)$ and truncated at $\pm 3$, where $\Sigma_x$ has an AR(1) structure with the  $(i,j)$-th entry being $0.5^{|i-j|}$. The true regression parameters $\beta^0$ are sparse. Its first element $\beta_1^0$ varies from 0 to 2 by an increment of 0.2, 
	four additional elements are assigned values of 1, 1, 0.3 and 0.3 with their positions randomly generated and then fixed for all of the simulations, and all other elements are zero. The underlying survival times $T_{ki}$ are simulated from an exponential distribution with hazard  $\lambda(t|X_{ki}) = \lambda_{0k} \exp\{X_{ki}^T\beta^0\}$, where $\lambda_{0k}$  are generated from  $\mathrm{Uniform}(0.5,1)$ and then  fixed throughout. 
	As in \citet{fang2017testing} and \citet{fan2002variable},  the censoring times $C_{ki}$'s are simulated independently from  an exponential distribution with hazard  $\lambda_c(t|X_{ki}) = 0.2 \lambda_{0k} \exp\{X_{ki}^T\beta^0\}$,   resulting in an overall censoring rate around 20\%. 
		
	For the lasso estimator, we use 5-fold within-stratum  cross-validation to select  $\lambda$.  
	In Scenarios 1--3 with small numbers of strata, each stratum serves as a cross-validation fold for the  selection of $\gamma$;  in Scenario 4 with 40 strata,  we perform 10-fold cross-validation as described in Algorithm \ref{algo1} and randomly assign 4 strata to each fold. For each parameter configuration, we simulate  100 datasets, based on which we compare estimation biases of  $\beta^0_1$, 95\% confidence interval coverage probabilities,  model-based standard errors, and empirical standard errors across the \lu{six} methods.

	 \begin{landscape}
		\begin{figure}[htb!]
			\centering
			\includegraphics[height=0.8\textheight]{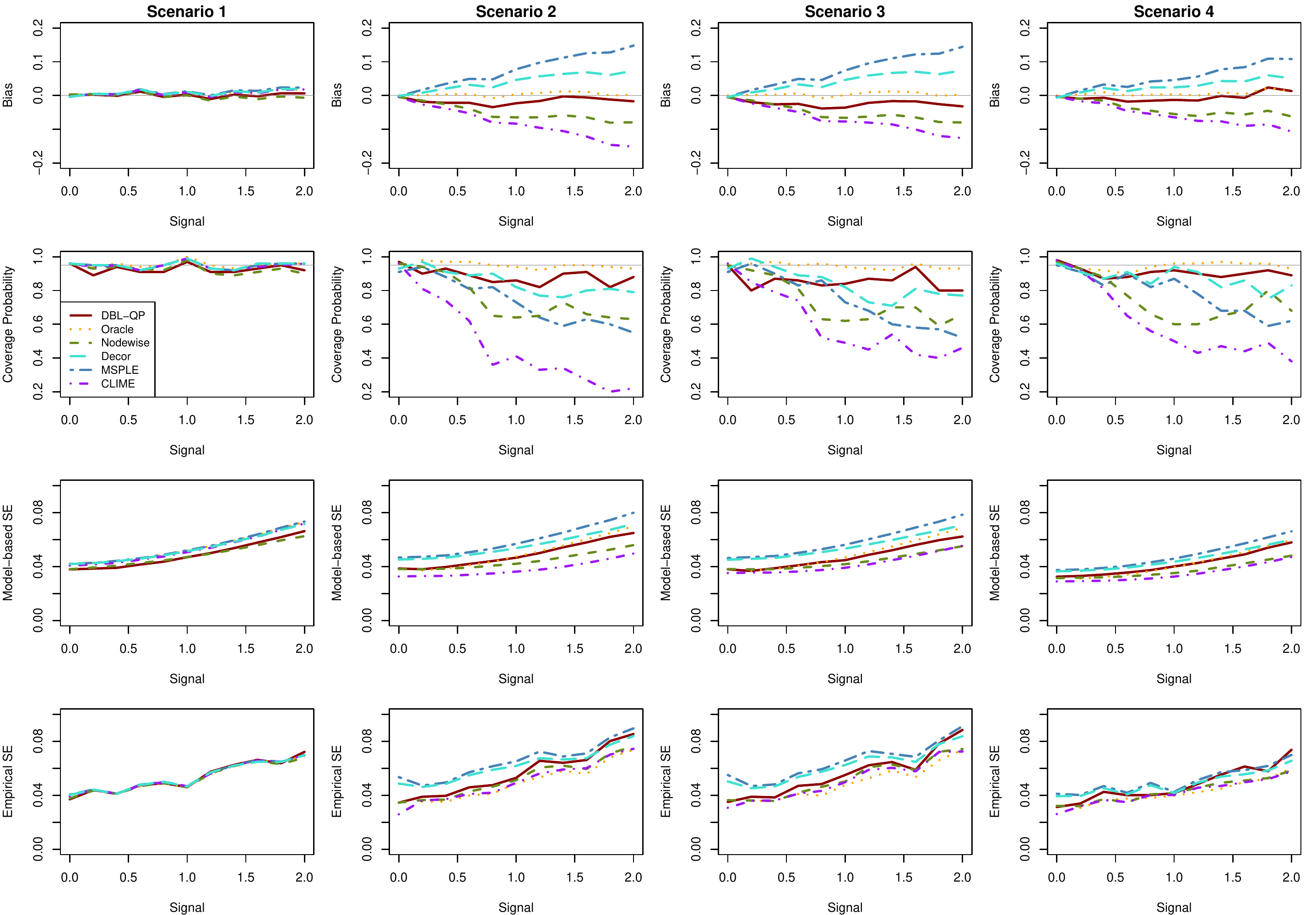}
			\caption{Estimation bias, 95\% coverage, model-based standard error, and empirical standard error for $\beta^0_1$ of six different methods.  Black horizontal lines are references to 0 for bias or 95\% for coverage probability.}
			\label{fig:4scenarios}
		\end{figure}
	\end{landscape}

	Figure \ref{fig:4scenarios} shows that, in Scenario 1 that features a small number of covariates ($p=10$), all \lu{six} methods perform well \lu{and similarly}; in Scenarios 2--4 with a relatively large number of covariates  ($p=100$), which is close to the number of covariates in the real data we will analyze,  \lu{our proposed  DBL-QP estimator well corrects the biases of the lasso estimates and maintains good confidence interval coverage (excluding the practically impossible ``Oracle'' estimator), but MSPLE, Nodewise, Decor and CLIME all present larger biases compared to DBL-QP as $\beta^0_1$ increases from 0 to 2.  CLIME, Nodewise and MSPLE have worse confidence interval coverage in general.
	As de-biased lasso methods, CLIME and Nodewise produce much smaller model-based standard error estimates,  which also contribute to their poor covarage probabilities. This is likely due to that both methods (CLIME and Nodewise) use penalized estimators for inverse information matrix estimation, and such penalization induces biases towards zero.}

	To recapitulate, the proposed  DBL-QP provides less biased estimates and better confidence interval coverage than \lu{the conventional MSPLE and three other competitors (Nodewise, Decor and CLIME adapted to the stratified setup)} when the sample size is moderate relative to the number of covariates,  although \lu{all methods} give almost identical results when $p$  is rather small. Hence, when $p < N$, our proposed  DBL-QP approach is at least as good as \lu{all the other methods}, and should be recommended for use.

	\section{Analysis of the SRTR kidney transplant data} 
	\label{sec:app}
	
	

\lu{The SRTR data set features 94 covariates from both donors and recipients, and the number of covariates is seen as relatively large for some recipient groups. With its reliable performance as demonstrated in simulations, we apply our  DBL-QP approach to analyze the SRTR data, while using MSPLE as a benchmark.}
The outcome is graft failure free survival,  the time from transplant to graft failure or death, whichever comes first. Our primary goal is to investigate the joint associations of these covariates with graft failure for three recipient groups   defined in  Table \ref{tab:character} separately. By simultaneously  considering all available donor and recipient covariates,
we aim to account for confounding  and provide asymptotically valid inference for the covariate effects, which differs from post hoc inference that only focuses on a smaller set of covariates selected by stepwise selection. The effect of donor age, in the presence of other risk factors,  is worth  investigating, as the debatable ``one-size-fit-all'' practice of donor-recipient age matching unfortunately is not suited for the benefit of transplantation \citep{keith2004effect,veroux2012age,dayoub2018effects}.

\subsection{Data details}

Included in our analysis are 9,195  recipients who received kidney-only transplants from deceased donors, had no prior solid organ transplants,  and were at least 18 years old at the time of transplantation during  2000 and 2001.  We focus  on those with these same cohort years in order to eliminate the cohort effect. Moreover, this group of patients had longer follow-up than those from the later cohort years. 
See Appendix \ref{appA} for a full list of included variables in the analysis.
 In the three receipts' age groups, respectively, the sample sizes are 3388, 4359 and 1448, the censoring rates are 53.1\%, 46.5\% and 30.0\%,  the median numbers of patients within each transplant center are 32, 31 and 27, and the restricted mean survival times by $13$ years are 9.1, 8.6 and 7.1 years.
	To select the tuning parameters, we implement 5-fold cross-validation   
	by randomly selecting one fifth of transplant centers without replacement as testing data and the rest as training data. 

	\subsection{Results} 
	
	We begin with examining the overall effect of donors' age on graft failure and testing the null hypothesis that, within each recipient group and after adjusting for the other risk factors,  all the donor age groups, i.e. $\le 10, (10, 20],
	(20,30], (30,40], (40,50], (50, 60]$ and $60+$, have the same risk of graft failure. Based  on Theorem \ref{thm:simul} and Corollary \ref{coro:thm2}, with $(m,p)=(6,94)$, we perform tests for the null contrasts, and the obtained statistics significantly reject the null hypotheses for all three recipient groups (within recipients aged 18-45: $\chi^2=40.4$, df=6, p-value=$3.9 \times 10^{-7}$; recipients aged 45-60: $\chi^2=34.5$, df=6, p-value=$5.3 \times 10^{-6}$; recipients aged over 60: $\chi^2=14.2$, df=6, p-value=$2.8 \times 10^{-2}$).
Indeed,	Figure \ref{fig:don_age}, which depicts the risk-adjusted effect of donors' age across the three recipient age groups,
	shows a general trend of increasing hazards for those receiving  kidneys from older donors, likely due to renal aging.  The estimates and confidence intervals obtained by our
	proposed  DBL-QP differ from those obtained by MSPLE, and the differences are the most obvious in the $60+$ year recipient group, which has the smallest sample size. As presented in our simulations, MSPLE may produce biased estimates with improper confidence intervals, especially when the sample size is relatively small.

	On the other hand, the  proposed DBL-QP  method may shed new light into the aging effect, which seems to be non-linear with respect to donors' age. First,
	using the results of Theorem \ref{thm:main} and Corollary \ref{coro:ci}, our tests detect  no significant   differences in hazards between those receiving kidneys from donors aged under 10 or $(10, 20]$ and  $(20, 30]$ (reference level) years old, within all the three recipient age groups.
	Second, significantly increased hazards are observed as early as when donors' age reached 30-40, as  compared to the reference level of $(20, 30]$, in the 18-45 years old recipient group, with an estimated hazard ratio (HR) of 1.16 (95\% CI: 1.01--1.34, p-value=$4.1\times 10^{-2}$). In contrast, there are no	significant differences between receiving organs from  $(30,40]$ years old donors and the reference level of $(20, 30]$, among the 
	45-60 years old recipients
	(HR= 0.96, 95\% CI: 0.85--1.09, p-value=$5.1\times 10^{-1}$)
and the	 $60+$ years old recipients  (HR=1.07, 95\% CI: 0.88--1.30, p-value=$5.0\times 10^{-1}$).
	Third,  kidneys from $60+$ years old donors confer the highest hazards, with the estimated risk-adjusted HRs (compared to the reference level $(20, 30]$) being 1.83 (95\% CI: 1.48--2.28, p-value=$4.3 \times 10^{-8}$), 1.40 (95\% CI: 1.21--1.61, p-value=$4.1\times 10^{-6}$) and  1.37 (95\% CI: 1.14--1.63, p-value=$5.2 \times 10^{-4}$) among the three recipient age groups respectively. 
	This means that, compared to the older recipients, recipients of  18-45 years old tend to experience a greater hazard of graft failure when receiving kidneys from donors over 60 years old.
	Caution needs to be exercised when allocating  kidneys from older donors to young patients \citep{lim2010donor,kabore2017age,dayoub2018effects}.

		\begin{figure}
		\centering
		\includegraphics[width=\textwidth]{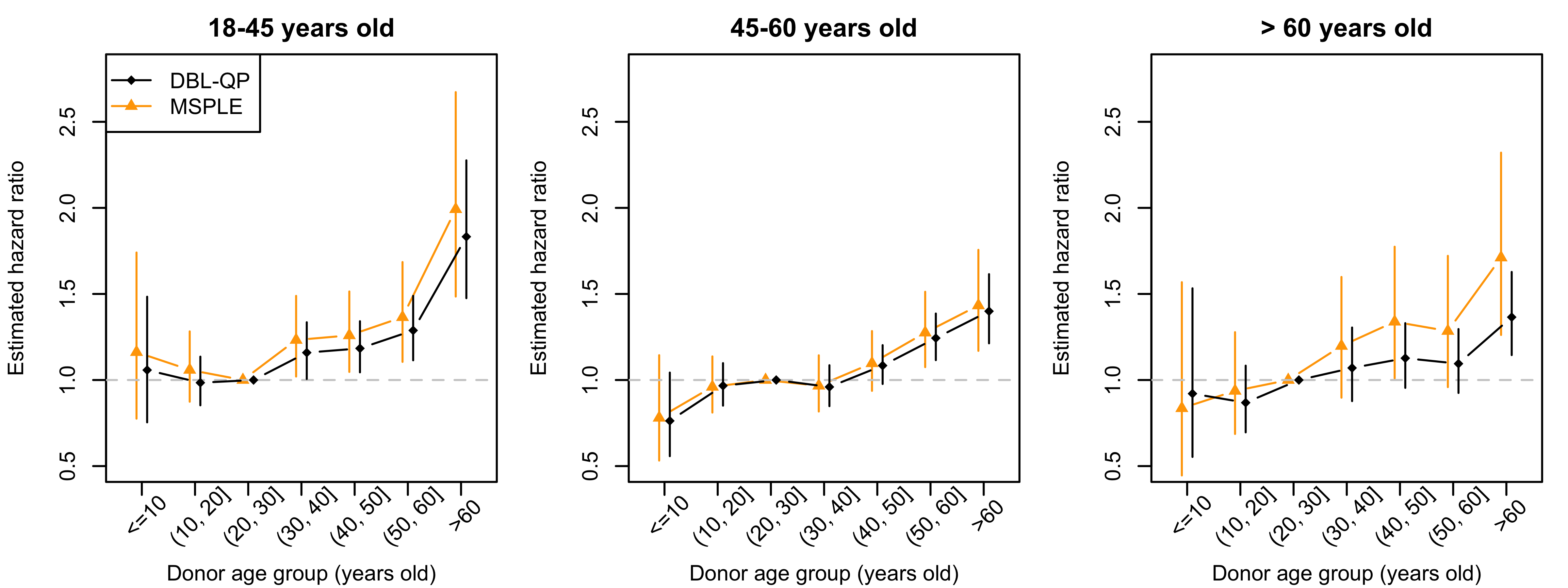}
		\caption{Estimated hazard ratios and the corresponding $95\%$ confidence intervals of different donor age categories with reference to the $(20, 30]$ donor age category, after adjusting for all other variables, in three recipient groups. 
		}
		\label{fig:don_age}
	\end{figure}

	Our method also delineates the associations of  clinical indicators with  graft failure,  provides more reliable inference,  and compares the relative strengths across recipient age groups.
\lu{By naively applying lasso, 64, 44 and 27 covariates are selected with non-zero coefficients in the 18-45, 45-60, and 60$+$ years old recipient groups, respectively.  In contrast,  the proposed DBL-QP identifies 22,  22 and 14 significant covariates in these three recipient groups, respectively,  from rigorous hypothesis tests with size 0.05 based on the asymptotic distribution.}
	Figure \ref{fig:ci_three} shows the estimated coefficients and their 95\% confidence intervals for covariates that are significant at level 0.05 in at least one recipient group. We highlight several noteworthy results.
	
	First,  recipients' primary kidney diagnosis  plays a critical role in kidney graft failure \citep{wolfe1991survival}. Compared to recipients with primary diagnosis of 
	diabetes (the reference level), those with polycystic kidneys (variable 2 in Figure \ref{fig:ci_three}) have a reduced risk of graft failure, with highly significant lower HRs of  0.54 (95\% CI: 0.42--0.70, p-value=$3.6 \times 10^{-6}$), 0.65 (95\% CI: 0.57--0.75, p-value=$4.4 \times 10^{-9}$) and   0.74 (95\% CI: 0.60--0.92, p-value=$5.3 \times 10^{-3}$) for the three age groups respectively.
	Compared to diabetes, 
	primary diagnosis of 
	glomerular disease (variable 26 in Figure \ref{fig:ci_three}) is significantly associated with a reduced risk of  graft failure only in the $60+$ years old recipient group (HR=0.79, 95\% CI: 0.66--0.96, p-value=$1.4 \times 10^{-2}$), and primary diagnosis of  hypertensive  nephrosclerosis (variable 29 in Figure \ref{fig:ci_three}) is
	significantly associated with a higher hazard of graft failure only in the 45-60 years old recipient group  (HR=1.12, 95\% CI: 1.01--1.23, p-value=$2.5 \times 10^{-2}$).

	Second, since diabetes is the most prevalent among end-stage renal patients \citep{kovesdy2010glycemic}, we code recipients' diabetic status at transplant as non-diabetic (reference level), diabetic for 0-20 years (variable 13 in Figure \ref{fig:ci_three}), and 20+ years (variable 3 in Figure \ref{fig:ci_three}).
	Our stratified analysis reveals that  diabetics is a stronger risk factor for young recipients aged between 18 and 45 years old than for older recipients, regardless of duration of diabetes.
	
	Third, instead of using the total number of mismatches as done in the literature, we consider the number of mismatches separately for each 
	HLA locus for more precisely pinpointing the effects of  mismatching loci. Our results reveal that the HLA-DR mismatches (variable 9 in Figure \ref{fig:ci_three})  are more strongly associated with graft failure than the HLA-A (variable 18 in Figure \ref{fig:ci_three}) and HLA-B mismatches (non-significant in any recipient group), which  are consistent with a meta-analysis based on 500,000 recipients \citep{shi2018impact}. 
	
	
	Finally, to study the granular impact of  recipient age on graft failure \citep{karim2014recipient},  we treat  recipient age (divided by 10) as a continuous variable (variable 4 in Figure \ref{fig:ci_three}) in the model within each recipient age group.  Interestingly, we find that
	increasing age is associated with a higher hazard 
	in the two older recipient groups (HR=1.31, 95\% CI: 1.19--1.44, p-value=$1.3\times 10^{-8}$, for recipients aged 45-60; HR=1.22, 95\% CI: 1.07--1.40, p-value=$3.6\times 10^{-3}$, for recipients aged 60+),
	but with  a lower hazard of graft failure  in the 18-45 recipient age group (HR=0.89, 95\% CI: 0.83--0.95, p-value=$5.2 \times 10^{-4}$).  
	This is likely because that younger patients  generally had poorer adherence to treatment, resulting in higher risks of graft loss \citep{kabore2017age}. The results also reinforce the necessity of separating analyses for different recipient age groups.
	
	

\begin{figure}
    \centering
    \includegraphics[width=0.82\textwidth]{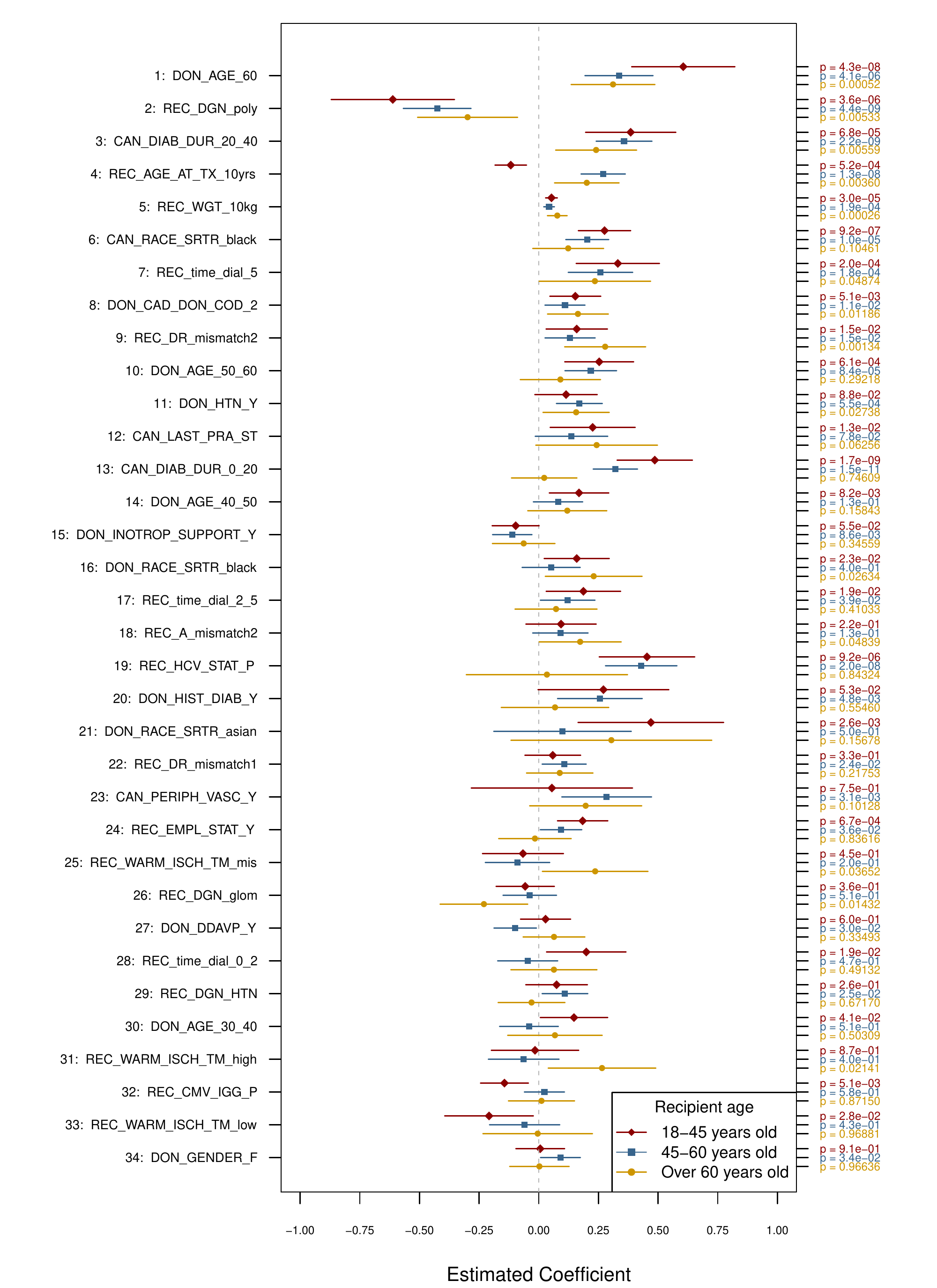}
    \caption{Estimated regression coefficients in the stratified Cox models using the proposed  DBL-QP method, and the corresponding $95\%$ confidence intervals, presented  by recipient age group. The covariates included are significant at level 0.05 in at least one recipient group, after adjusting for all other covariates.}
    \label{fig:ci_three}
\end{figure}

As a side note, we compare DBL-QP and  MSPLE in the estimated coefficients and standard errors.
Figure \ref{fig:QPvsMPLE_all} shows that in the 45-60 age group with the largest number of subjects, the point estimates obtained by the two methods almost coincide with each other,    whereas in the  $60+$ age group with the smallest sample size,  MSPLE tends to have larger absolute estimates than the de-biased lasso.
Moreover, the standard errors estimated by MSPLE are likely to be larger than those by  our method across all the age groups. 
These observations agree with the results of our simulations (Scenarios 2--4), which show that MSPLE yields large biases in estimated coefficients and standard errors, especially when the sample size is relatively small, whereas our proposed  DBL-QP method  draws more valid inferences by maintaining proper type I errors and coverage probabilities.
	
		\begin{figure}
		\centering
		\includegraphics[width=\textwidth]{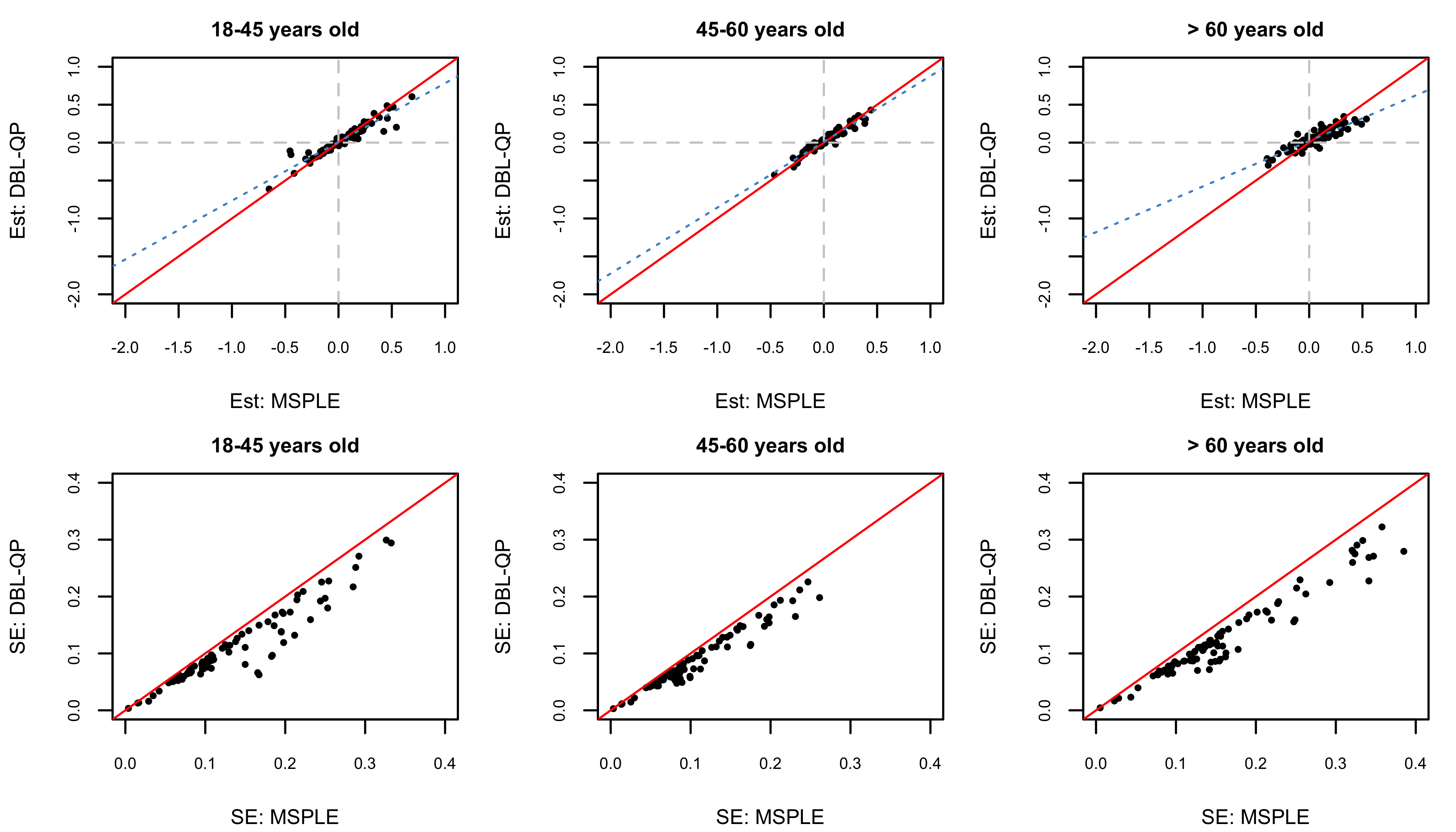}
		\caption{Comparison between the coefficient estimates (top) and the model-based standard errors (bottom) by the de-biased lasso (DBL-QP) and the maximum stratified partial likelihood estimation (MSPLE) in three recipient age groups. The solid red lines are 45-degree reference lines, and the dotted blue lines in the top figures represent the fitted linear regression of the DBL-QP estimates on the MSPLE estimates.}
		\label{fig:QPvsMPLE_all}
	\end{figure}

	\section{Concluding remarks} 
	\label{sec:discuss}
	
	The work is motivated by 
	an urgent call of  better understanding the complex mechanisms behind post-kidney transplant graft failure. Our modeling framework is 
	Cox models stratified by transplant centers,  due to their  strong confounding effects  on graft failure. 
	To  adjust for confounders to the extent possible, we have included an extended list of 94 covariates from recipients and donors, which has not been done in the literature.  A particular scientific question to address is the debatable donor-recipient age matching criterion in kidney transplantation. Fitting separate  models by recipient age enables direct assessments of the donor age effects in different recipient age groups, which  differs from using donor-recipient age difference as in \citet{ferrari2011effect}.
	Specifically,  we have followed a common practice of fitting separate models in age groups of 18-45, 45-60 and $60+$ years.  The commonly used MSPLE yielded biased estimates and  unreliable inference in some smaller age groups, though the samples outnumbered the covariates. In particular, the $60+$ years recipient group had only 1448 recipients in 43 different transplant centers, and MSPLE yielded more dramatic estimates for those donor age effects of over 30 years old (Figure \ref{fig:don_age}).  Our simulation results    also confirmed such a problematic phenomenon. Therefore, a statistical method that can guarantee reliable estimates and valid inference is much needed for delineating the associations of interest with graft failure when the number of covariates is relatively large in stratified Cox models. 
	
	Inspired by the de-biased lasso method for linear regression \citep{javanmard2014confidence}, we have developed a de-biased lasso approach via quadratic programming for stratified Cox models. \lu{Despite progress made in high-dimensional inference for Cox models, virtually no work has considered stratified settings, theoretically or empirically. We have shown that
	in the ``large $N$, diverging $p$'' scenario, our approach possesses desirable 
	asymptotic properties and finite-sample performance, and is more suitable for the analysis of the SRTR data than the competing methods illustrated in our simulation studies.
	Computationally,  based on a previous work on Cox models without stratification \citep{xia2021cox}, for the estimation of $\Thetabeta$, the computational speed using \texttt{solve.QP} in R was much faster than that using the R packages \texttt{clime} or \texttt{flare} adopted by \citet{yu2018confidence}.  }

	 Applications of our method to the SRTR data generated new biological findings.  After categorizing donors' age and controlling for other risk factors listed in Appendix \ref{appA}, we find that {organs from  older donors} are associated with an increased hazard of graft failure and that the dependence on donors' age is non-linear: within the youngest recipient group (18-45 years),  significant differences from the reference donor age category (20-30 years) were detected as early as when donors reached 30-40 years old, whereas significant differences were detected only when donors reached 50-60 or 60+ years  
	 within the two older recipient groups, respectively; in other words, receiving kidneys from older or younger donors, such as $60+$ versus 20-30 years, presented larger differences than in the other two recipient groups. These results, which were not reported in the literature, may provide new empirical evidence to aid stake-holders, such as patients, families, physicians and  policy makers, in decisions  on donor-recipient age matching.  
	
	\lu{A few technical points are noteworthy.  First,} our work deals with the ``large $N$, diverging $p$'' scenario, as embedded in the motivating data, and approximates $\Theta_{\beta^0}$ via quadratic programming without positing any sparsity conditions on $\Theta_{\beta^0}$. This distinguishes from the related  literature \citep{fang2017testing,yu2018confidence,kong2018high} \lu{in the ``large $p$, small $N$'' scenario} that relies upon sparsity conditions on the numbers of non-zero elements in the rows of $\Theta_{\beta^0}$, which \lu{are hardly discussed in depth and may not hold nor have explicit interpretations for Cox models}.  \lu{For example, when the rows of $\Theta_{\beta^0}$ are not sparse, our dimension requirement for $p$ is less stringent than in \citet{yu2018confidence}, by a factor of $\sqrt{\log(Np)}$.  } \lu{Moreover, when $p > N$, the several de-biased methods aforementioned may not yield reliable inference results, as empirically $\Theta_{\beta^0}$  cannot be estimated well, and biases in the lasso estimator are often not sufficiently corrected for in this scenario for Cox models. New approaches, such as sample-splitting approaches \citep{fei2021estimation}
	that  bypass the estimation of $\Theta_{\beta^0}$,  can be consulted.}
	
\lu{Second, tuning parameter selection is critical in high-dimensional inference. Our proposed method deploys a single tuning parameter $\gamma$ for de-biasing the estimates of all $\beta_j$'s. This is a computationally feasible and commonly adopted strategy, presenting a satisfactory performance in our numerical studies, and can  be extended to adapt to the variability of individual coefficient estimation. For example, one may consider the following estimation procedure for the $j$th row of $\hat{\Theta}$ along the line of adaptive CLIME \citep{cai2016estimating}:
\[
\min_m \{ m^T \hat{\Sigma} m: | (\hat{\Sigma} m - e_j)_k | \le \gamma_{jk},  k = 1, \ldots, p \}.
\]
Here, $\gamma_{jk}$'s are supposed to be adaptively estimated through a carefully designed procedure. However, the design of such an appropriate procedure requires complicated theoretical analysis in Cox models, unstratified or stratified, to determine the desirable rates of $ \gamma_{jk}$'s, among other tasks. Given that such complexity is beyond the scope of this paper, we will not pursue this route here in details but will leave it for future research.}

\lu{ Third,  though primarily focusing on the associations between the risk factors and survival (through Theorem \ref{thm:main}), the proposed method can  be used for patient risk scoring and conditional survival probability estimation. For example, the de-biased estimates may be plugged into the Breslow's estimator \citep{kalbfleisch2002statistical} for stratum-specific baseline hazards. The conditional survival probability estimation may not go beyond the time point $\tau$ due to censoring.  }

\lu{ Lastly, we use Cox models stratified by transplant centers to account for but avoid explicitly modeling the center effects. Alternatively, random effects models can be used for clustered survival data analysis; for example, \citet{vaida2000proportional} generalized the usual frailty model to allow multivariate random effects. However, in a random effects model,  the distribution of random effects needs to be specified, and  the coefficients only have conditional interpretations, given a cluster. We may pursue this elsewhere.
}
	
\lu{We have implemented the proposed DBL-QP method with cross-validation in R and Rcpp, which is available online at \href{https://github.com/luxia-bios/StratifiedCoxInference/}{https://github.com/luxia-bios/StratifiedCoxInference/} with simulated examples.}

\section*{Appendix}
\begin{appendix}
\section{SRTR data}\label{appA}
The SRTR dataset analyzed in this article can be accessed by applying through the OPTN website \href{https://optn.transplant.hrsa.gov}{https://optn.transplant.hrsa.gov}. The interpretation and reporting of the SRTR data results are solely the responsibility of the authors and should not be viewed as official opinions of the SRTR or the United States Government.
	
	The 94 covariates, including dummy variables, are derived from the following factors. Donor factors include: ABO blood type, age, cytomegalovirus antibody, hepatitis C virus antibody, cause of death, cardiac arrest since event leading to declaration of death, serum creatinine, medication given to donor (DDAVP, dopamine and dobutamine), gender, height,  history of cancer, cigarette smoking, history of drug abuse, hypertension,  diabetes, inotropic support, inotropic agents at time of incision, non-heart beating donor, local or shared organ transplant, race, and weight.
	Recipient factors include: ABO blood type, history of diabetes and duration, angina/coronary artery disease, symptomatic peripheral vascular disease, drug treated systemic hypertension, drug treated COPD, gender (and previous pregnancies for females), sensitization (whether peak and/or current panel-reactive antibodies exceed 20\%), previous malignancy, peptic ulcer disease, symptomatic cerebrovascular disease, race, total serum albumin,  age at transplant, number of HLA mismatches (A, B and DR), cytomegalovirus status, total cold ischemic time, primary kidney diagnoses, pre-transplant dialysis and duration, the Epstein–Barr virus serology status, employment status, hepatitis B virus status, hepatitis C virus status, height, pre-implantation kidney biopsy, pre-transplant blood transfusions, transplant procedure type, warm ischemic time and weight.

\section{Regularity conditions}\label{appB}
\renewcommand{\theassumption}{B.\arabic{assumption}}

Assumptions \ref{assump1}--\ref{assump5} below ensure that  Theorem \ref{thm:main} hold.
	
	\begin{assumption}\label{assump1}
		Covariates are almost surely uniformly bounded, i.e. $\| X_{ki} \|_{\infty} \le M$ for some positive constant $M<\infty$ for all $k$ and $i$.
	\end{assumption}
	
	\begin{assumption}
		\label{assump2}
		$| X_{ki}^T \beta^0 | \le M_1$ uniformly for all $k$ and $i$ with some positive constant $M_1 < \infty$ almost surely.
	\end{assumption}
	
	\begin{assumption}
		\label{assump3}
		The follow-up time stops at a finite time point $\tau > 0$, with probability $\pi_0 = \min_{k} P (Y_{ki} \ge \tau) > 0$.
	\end{assumption}

\begin{assumption}
		\label{assump4}
		 \lu{ For any $t\in [0, \tau]$, 
		\[
		\frac{c^T \Theta_{\beta^0}}{c^T \Theta_{\beta^0} c}  \left[ \sum_{k=1}^K r_k \int_0^t  \left\{ \mu_{2k}(u; \beta^0) - \frac{\mu_{1k}(u; \beta^0) \mu_{1k}(u; \beta^0)^T}{\mu_{0k}(u; \beta^0)} \right\} \lambda_{0k}(u) d u \right] \Theta_{\beta^0} c \rightarrow v(t; c)
		\]
		as  $n \rightarrow \infty$ for some function $v(t; c) > 0$ of $t$ that also depends on the choice of $c$. }
	\end{assumption}  
	
	\begin{assumption}
		\label{assump5}
		\lu{There exists a constant $\epsilon_0>0$ such that $\lambda_{\mathrm{min}}(\Sigma_{\beta^0}) \ge \epsilon_0$, where $\lambda_{\mathrm{min}} (\cdot)$ is the smallest eigenvalue of a matrix. 
         }
	\end{assumption}
	

	For inference on multiple linear combinations or contrasts as described in Theorem \ref{thm:simul}, Assumption \ref{assump4} needs to be replaced with the following Assumption \ref{assump6}, which is a multivariate version of Assumption \ref{assump4}.
	
	\begin{assumption}
		\label{assump6}
		\lu{ For any $\omega \in \mR^{m}$ and any $t\in [0, \tau]$, 
		\[
		\frac{\omega^T J \Theta_{\beta^0} }{\omega^T J  \Theta_{\beta^0} J^T \omega}  \left[ ~ \sum_{k=1}^K r_k \int_0^t  \left\{ \mu_{2k}(u; \beta^0) - \frac{\mu_{1k}(u; \beta^0) \mu_{1k}(u; \beta^0)^T}{\mu_{0k}(u; \beta^0)} \right\} d\Lambda_{0k}(u) \right] \Theta_{\beta^0} J^T \omega 
		\]
	converges to $v^{\prime}(t; \omega, J)$	as  $n \rightarrow \infty$, for some function $v^{\prime}(t; \omega, J) > 0$ of $t$, that also depends on the choice of $\omega$ and $J$. }
	\end{assumption}

 \lu{It is common in the literature of high-dimensional inference to assume bounded covariates as in Assumption B.1. \citet{fang2017testing} and \citet{kong2018high} also posed Assumption B.2 for Cox models, i.e. uniform boundedness on the multiplicative hazard. Under Assumption B.1, Assumption B.2 can be implied by the bounded overall signal strength $\| \beta^0 \|_1$. Assumption B.3 is a common assumption in survival analysis \citep{andersen1982cox}. Assumption B.4
and its multivariate version, Assumption B.6,  ensure the convergence of the variation process, which is key in applying the martingale central limit theorem. They are less stringent comparing to the boundedness assumption on $\| \Theta_{\beta^0} X_{ki} \|_{\infty}$ that is equivalent to the assumptions for statistical inference in \citet{van2014asymptotically} on high-dimensional generalized linear models and in \citet{fang2017testing} on high-dimensional Cox models.  The boundedness of the smallest eigenvalue of $\Sigma_{\beta^0}$ in Assumption B.5 is common in inference for high-dimensional models \citep{van2014asymptotically,kong2018high}.  
Since we focus on random designs, unlike \citet{huang2013oracle},  \citet{yu2018confidence} and \citet{fang2017testing}, we do not directly assume the compatibility condition on $\ddot{\ell}(\beta^0)$;  instead, we impose Assumption B.5 on the population-level matrix $\Sigma_{\beta^0}$, which leads to the compatibility condition for a given data set with probability going to one.}

\end{appendix}

\section*{Acknowledgements}
This work was supported in part by NIH Grants (R01AG056764, R01AG075107, R01CA249096 and U01CA209414)  and NSF Grant (DMS-1915711). 
\bibliographystyle{apalike} 
\bibliography{References3.bib}       


\newpage

\setcounter{section}{0}
\renewcommand{\thesection}{S.\arabic{section}}
\renewcommand{\theequation}{S.\arabic{equation}}
\renewcommand{\thelemma}{S.\arabic{lemma}}

\begin{center}
	{\bf \Large
	Supplement to ``De-biased lasso for stratified Cox models with application to the national kidney transplant data"} \\[1em]
{\large Lu Xia\textsuperscript{a}, Bin Nan\textsuperscript{b}, and Yi Li\textsuperscript{c} \\[1em] }
	{\small 
	\textsuperscript{a}Department of Biostatistics, University of Washington, Seattle, WA \\ \textsuperscript{b}Department of Statistics, University of Californina, Irvine, CA \\ \textsuperscript{c}Department of Biostatistics, University of Michigan, Ann Arbor, MI }
\end{center}

For completeness of presentation, we first provide some useful  lemmas and their proofs, and then give the proofs of the main theorems in this supplementary material. Since corollaries are direct results of the corresponding theorems, their proofs are straightforward  and thus omitted. 

\section{Technical Lemmas}

For the $i$th subject in the $k$th stratum, define the counting process $N_{ki}(t) = 1(Y_{ki} \le t, \delta_{ki} = 1)$. The corresponding intensity process $A_{ki}(t; \beta) = \int_0^t 1 (Y_{ki} \ge s) \exp(X_{ki}^T \beta) d\Lambda_{0k}(s)$, where $\Lambda_{0k}(t) = \int_{0}^{t} \lambda_{0k}(s) ds$ is the baseline cumulative hazard function for the $k$th stratum, $k=1, \cdots, K$, $i = 1, \cdots, n_k$. Let $M_{ki}(t; \beta) = N_{ki}(t) - A_{ki}(t; \beta)$, then $M_{ki}(t; \beta^0)$ is a martingale with respect to the filtration $\mathcal{F}_{ki}(t) = \sigma \{ N_{ki}(s), 1(Y_{ki} \ge s), X_{ki}: s \in (0, t] \}$. 

Recall that the stratum-specific weighted covariate process $\heta_k(t; \beta) = \hmu_{1k} (t; \beta) / \hmu_{0k} (t; \beta)$, where $\hmu_{rk} (t; \beta) = (1/ n_k) \sum_{i=1}^{n_k} 1(Y_{ki} \ge t) X_{ki}^{\otimes r} \exp\{X_{ki}^T \beta\}$. Their population-level counterparts are $\mu_{rk} (t; \beta) = \E [1(Y_{k1} \ge t) X_{k1}^{\otimes r} \exp\{X_{k1}^T \beta\}]$ and $\eta_{k0}(t; \beta) = \mu_{1k} (t; \beta) / \mu_{0k} (t; \beta)$,  $r = 0, 1, 2$, $k=1, \cdots, K$. It is easily seen that the process $\{X_{ki} - \heta_k(t; \beta^0)\}$ is predictable with respect to the filtration $\mathcal{F}(t) = \sigma \{ N_{ki}(s), 1(Y_{ki} \ge s), X_{ki}: s \in (0, t], k=1, \cdots, K, ~ i=1, \cdots, n_k \}$.

\begin{lemma} \label{chap4:lemma:mom}
	Under Assumptions B.1--B.3, for $k = 1, \cdots, K$, we have
	\begin{align*}
		& \sup_{t \in [0, \tau]} | \hmu_{0k}(t; \beta^0) - \mu_{0k}(t; \beta^0) |  = \OP(\sqrt{\log(p) / n_k}), \\
		& \sup_{t \in [0, \tau]} \| \hmu_{1k}(t; \beta^0) - \mu_{1k}(t; \beta^0) \|_{\infty}  = \OP(\sqrt{\log(p)/n_k}), \\
		& \sup_{t \in [0, \tau]} \| \heta_k(t; \beta^0) - \eta_{k0}(t; \beta^0)  \|_{\infty}  = \OP(\sqrt{\log(p)/n_k}).
	\end{align*}
\end{lemma}

Lemma \ref{chap4:lemma:mom} is the extension of the results of Lemma A1 in \citet{xia2021cox} to each of the $K$ strata. We omit its proof here.

\begin{lemma} \label{chap4:lemma:lead}
	Assume $ \| \Theta_{\beta^0} \|_{1,1}^2 \{ \sqrt{\log(p) / n_{min}} + \max_k | n_k/N - r_k | \} \to 0$, and $\{ c^T \Thetabeta c \}^{-1} = \mathcal{O}(1)$. Under Assumptions B.1--B.5, for any $c \in \mathbb{R}^p$ such that $\| c \|_2 = 1$ and $\| c \|_1 \le a_*$ with some absolute constant $a_* < \infty$, we have
	\[
	\displaystyle \frac{ \sqrt{N} c^T \Thetabeta \dot{\ell}(\beta^0)}{\sqrt{c^T \Thetabeta c}} \overset{\mathcal{D}}{\rightarrow} N(0,1).
	\]
\end{lemma}

\begin{proof}[\textbf{Proof of Lemma \ref{chap4:lemma:lead}}]
	We rewrite
	\begin{align}  	\label{chap4:eq:clt_decomp} 
		\displaystyle \frac{- \sqrt{N} c^T \Thetabeta \dot{\ell}(\beta^0)}{\sqrt{c^T \Thetabeta c}}  & =  \displaystyle \frac{1}{\sqrt{N}} \sum_{k=1}^K \sum_{i=1}^{n_k} \frac{c^T \Thetabeta}{\sqrt{c^T \Thetabeta c}}  \left\{  X_{ki} - \frac{\hmu_{1k}(Y_{ki}; \beta^0)}{\hmu_{0k}(Y_{ki}; \beta^0)}  \right\} \delta_{ki} \nonumber \\
		& =  \frac{1}{\sqrt{N}} \sum_{k=1}^K \sum_{i=1}^{n_k} \int_{0}^{\tau} \frac{c^T \Thetabeta}{\sqrt{c^T \Thetabeta c}}  \left\{  X_{ki} - \frac{\hmu_{1k}(t; \beta^0)}{\hmu_{0k}(t; \beta^0)}  \right\} dN_{ki}(t) \nonumber \\
		& = \frac{1}{\sqrt{N}} \sum_{k=1}^K \sum_{i=1}^{n_k} \int_{0}^{\tau} \frac{c^T \Thetabeta}{\sqrt{c^T \Thetabeta c}}  \left\{  X_{ki} - \frac{\hmu_{1k}(t; \beta^0)}{\hmu_{0k}(t; \beta^0)}  \right\} dM_{ki}(t). 
	\end{align}
	Denote $\displaystyle U(t) = \frac{1}{\sqrt{N}} \sum_{k=1}^K \sum_{i=1}^{n_k} \int_{0}^{t} \frac{c^T \Thetabeta}{\sqrt{c^T \Thetabeta c}}  \left\{  X_{ki} - \frac{\hmu_{1k}(s; \beta^0)}{\hmu_{0k}(s; \beta^0)}  \right\} dM_{ki}(s)$. Then the variation process for $U(t)$ is
	\begin{align}
		\langle U \rangle (t) & \displaystyle = \sum_{k=1}^K \sum_{i=1}^{n_k} \frac{1}{N} \int_0^t \frac{c^T \Theta_{\beta^0}}{c^T \Theta_{\beta^0} c} \left\{  X_{ki} - \heta_k(u; \beta^0)  \right\}^{\otimes 2} 1(Y_{ki} \ge u) e^{X_{ki}^T \beta^0} d\Lambda_{0k}(u) \Theta_{\beta^0} c  \nonumber \\
		& = \displaystyle  \frac{c^T \Theta_{\beta^0}}{c^T \Theta_{\beta^0} c} \left[ \sum_{k=1}^K \frac{n_k}{N} \int_0^t \left\{  \hmu_{2k}(u; \beta^0) - \frac{\hmu_{1k}(u; \beta^0) \hmu_{1k}^T(u; \beta^0)}{\hmu_{0k}(u; \beta^0)}  \right\} d\Lambda_{0k}(u) \right] \Theta_{\beta^0} c.
	\end{align}
	Following the proof of Lemma A2 in \citet{xia2021cox}, we have
	\begin{align*}
		& \left\|  \displaystyle \int_0^t \left\{ \hmu_{2k} (u; \beta^0) - \frac{\hmu_{1k}(u; \beta^0) \hmu_{1k}(u; \beta^0)^T}{\hmu_{0k}(u; \beta^0)} \right\} d \Lambda_{0k}(u)  -  \right. \\
		& \qquad \displaystyle \left. \int_0^t \left\{ {\mu}_{2k} (u; \beta^0) - \frac{{\mu}_{1k}(u; \beta^0) {\mu}_{1k}(u; \beta^0)^T}{{\mu}_{0k}(u; \beta^0)} \right\} d \Lambda_{0k}(u)  \right\|_{\infty} = \OP(\sqrt{\log(p) / n_k})
	\end{align*}
	uniformly for all $t \in [0, \tau]$. Then, 
	\begin{align*}
		& \displaystyle  \frac{c^T \Theta_{\beta^0}}{c^T \Theta_{\beta^0} c}  \left[ \sum_{k=1}^K \frac{n_k}{N} \int_0^t \left\{  \hmu_{2k}(u; \beta^0) - \frac{\hmu_{1k}(u; \beta^0) \hmu_{1k}^T(u; \beta^0)}{\hmu_{0k}(u; \beta^0)}  \right\} d\Lambda_{0k}(u) \right] \Theta_{\beta^0} c \\
		= & \frac{c^T \Theta_{\beta^0}}{c^T \Theta_{\beta^0} c}  \left[ \sum_{k=1}^K r_k \int_0^t \left\{  \mu_{2k}(u; \beta^0) - \frac{\mu_{1k}(u; \beta^0) \mu_{1k}^T(u; \beta^0)}{\mu_{0k}(u; \beta^0)}  \right\} d\Lambda_{0k}(u) \right] \Theta_{\beta^0} c ~ + ~ \\
		& \qquad \OP \left\{ \| c \|_1^2 \| \Thetabeta \|_{1,1}^2 \left( \max_k |n_k/N - r_k| + \sqrt{\log(p)/ n_{min}} \right) \right\}. 
	\end{align*}
	By Assumption B.4,  $\langle U \rangle (t) \rightarrow_{P} v(t; c) > 0, ~ t \in [0, \tau]$.
	
	For any $\epsilon > 0$, define $G_{ki}(u) = \displaystyle \frac{1}{\sqrt{N}} \frac{c^T \Thetabeta}{\sqrt{c^T \Thetabeta c}}  \left\{  X_{ki} - \frac{\hmu_{1k}(u; \beta^0)}{\hmu_{0k}(u; \beta^0)}  \right\}$ and the truncated process 
	$
	\displaystyle U_{\epsilon}(t) =  \sum_{k=1}^K \sum_{i=1}^{n_k} \int_{0}^{t} G_{ki}(u)  1(|G_{ki}(u)| > \epsilon) dM_{ki}(u).
	$
	The variation process of $U_{\epsilon}(t)$ is 
	\[
	\langle U_{\epsilon} \rangle (t) = \sum_{k=1}^K \sum_{i=1}^{n_k} \int_{0}^{t} G^2_{ki}(u)  1(|G_{ki}(u)| > \epsilon) dA_{ki}(u),
	\]
	where $d A_{ki}(u) = 1(Y_{ki} \ge u) e^{X_{ki}^T \beta^0} d\Lambda_{0k}(u)$. 
	Since \[
	|  \sqrt{N} G_{ki}(u) | \le a_* \| \Theta_{\beta^0} \|_{1,1} 2M  \{ c^T \Thetabeta c \}^{-1/2}  = \mathcal{O}(\| \Theta_{\beta^0} \|_{1,1}),
	\]
	then $1(|G_{ki}(u)| > \epsilon)  = 0$ almost surely as $ \| \Theta_{\beta^0} \|_{1,1} / \sqrt{N} \asymp \| \Theta_{\beta^0} \|_{1,1} / \sqrt{n_{min}} \to 0$. So $\langle U_{\epsilon} \rangle (t) \rightarrow_{P} 0$. By the martingale central limit theorem, we obtain the desirable result. 
\end{proof}

\begin{lemma} \label{chap4:lemma:lasso}
	Under Assumptions B.1--B.3 and B.5, for $\lambda \asymp \sqrt{\log(p)/n_{min}}$, the lasso estimator $\hbeta$ satisfies 
	\[
	\| \hbeta - \beta^0 \|_1 = \OP(s_0 \lambda), ~ \displaystyle \frac{1}{N} \sum_{k=1}^K \sum_{i=1}^{n_k} |X_{ki}^T (\hbeta - \beta^0)|^2 = \OP(s_0 \lambda^2).
	\]
\end{lemma}

\begin{proof}[\textbf{Proof of Lemma \ref{chap4:lemma:lasso}}]
	This result is adapted from the proof in \citet{kong2014non}, with  modifications as follows. An intermediate replacement for  the negative log-likelihood in the $k$th stratum 
	\[
	\ell^{(k)}(\beta) = \displaystyle - \frac{1}{n_k} \sum_{i=1}^{n_k} \left[ \beta^T X_{ki} - \log \left\{  \frac{1}{n_k} \sum_{j=1}^{n_k} 1(Y_{kj} \ge Y_{ki}) \exp(\beta^T X_{kj})  \right\} \right] \delta_{ki}
	\]
	can be defined as 
	\[
	\widetilde{\ell}^{(k)}(\beta) = - \displaystyle \frac{1}{n_k} \sum_{i=1}^{n_k} \left\{ \beta^T X_{ki} - \log \mu_{0k}(Y_{ki}; \beta) \right\}  \delta_{ki},
	\]
	which is a sum of $n_k$ independent and identically distributed terms.
	The target parameter is $$\bar{\beta} = \argmin_{\beta}  \displaystyle \E \left\{ \sum_{k=1}^K \frac{n_k}{N} \widetilde{\ell}^{(k)}(\beta) \right\}. $$ Then the excess risk for any given $\beta$ is
	\[
	\mathcal{E}(\beta) = \displaystyle \E \left\{ \sum_{k=1}^K \frac{n_k}{N} \widetilde{\ell}^{(k)}(\beta) \right\}  - \displaystyle \E \left\{ \sum_{k=1}^K \frac{n_k}{N} \widetilde{\ell}^{(k)}(\bar\beta) \right\}.
	\]
	We refer remaining details to \citet{kong2014non}.
\end{proof}

\begin{lemma} \label{chap4:lemma:const_sol}
	Under Assumptions B.1--B.3 and B.5 and assume $\lambda \asymp \sqrt{\log(p)/n_{min}}$,  it holds with probability going to one that $\| \Theta_{\beta^0} \hSigma - I_p \|_{\infty} \le \gamma$, with $\gamma \asymp \| \Theta_{\beta^0} \|_{1,1} \{ \max_{1\le k \le K} | n_k/N - r_k | + s_0 \lambda \}$.
\end{lemma}

Lemma \ref{chap4:lemma:const_sol} shows that when $\gamma \asymp \| \Theta_{\beta^0} \|_{1,1} \{ \max_{1\le k \le K} | n_k/N - r_k | + s_0 \lambda \}$, the $j$th row of $\Theta_{\beta^0}$ ($j=1, \cdots, p$) will be feasible in the constraint of the corresponding quadratic programming problem  with probability going to one.

\begin{proof}[\textbf{Proof of Lemma \ref{chap4:lemma:const_sol}}]
	We first derive the rate for $\| \hSigma - \Sigma_{\beta^0} \|_{\infty}$. Note that 
	\begin{align*}
		& \hspace{-0.4in} \| \hSigma - \Sigma_{\beta^0} \|_{\infty} \\
		\le & \ \left\| \displaystyle \frac{1}{N} \sum_{k=1}^K \sum_{i=1}^{n_k} \int_{0}^{\tau} \left[  \{ X_{ki} - \heta_k(t; \hbeta) \}^{\otimes 2} -  \{ X_{ki} - \eta_{k0}(t; \beta^0) \}^{\otimes 2}  \right] dN_{ki}(t) \right\|_{\infty} \\
		& \quad + \left\| \displaystyle \frac{1}{N} \sum_{k=1}^K \sum_{i=1}^{n_k} \int_{0}^{\tau} \{ X_{ki} - \eta_{k0}(t; \beta^0) \}^{\otimes 2} dN_{ki}(t) - \Sigma_{\beta^0} \right\|_{\infty} \\
		\equiv & \  a_{N1} + a_{N2}. 
	\end{align*}
	By the boundedness Assumption B.1, Lemma \ref{chap4:lemma:mom} and Lemma \ref{chap4:lemma:lasso},
	\begin{align*}
		a_{N1} \le & \left\| \displaystyle \frac{1}{N} \sum_{k=1}^K \sum_{i=1}^{n_k} \int_{0}^{\tau} \{ X_{ki} - \heta_{k}(t; \hbeta) \} \{ \eta_{k0}(t; \beta^0) - \heta_{k}(t; \hbeta) \}^T dN_{ki}(t)  \right\|_{\infty} \\
		& + \left\| \displaystyle \frac{1}{N} \sum_{k=1}^K \sum_{i=1}^{n_k} \int_{0}^{\tau} \{ \eta_{k0}(t; \beta^0) - \heta_{k}(t; \hbeta) \} \{ X_{ki} - \eta_{k0}(t; \beta^0) \}^T dN_{ki}(t)  \right\|_{\infty} \\ 
		\le &  \displaystyle \frac{4M}{N} \sum_{k=1}^K \sum_{i=1}^{n_k} \int_{0}^{\tau} \| \eta_{k0}(t; \beta^0) - \heta_{k}(t; \hbeta) \|_{\infty} dN_{ki}(t) \\
		\le & \displaystyle \frac{4M}{N} \sum_{k=1}^K \sum_{i=1}^{n_k} \int_{0}^{\tau} \| \eta_{k0}(t; \beta^0) - \heta_{k}(t; \beta^0) \|_{\infty} dN_{ki}(t) \\
		& + \displaystyle \frac{4M}{N} \sum_{k=1}^K \sum_{i=1}^{n_k} \int_{0}^{\tau} \| \heta_{k}(t; \beta^0) - \heta_{k}(t; \hbeta) \|_{\infty} dN_{ki}(t) \\
		\le & 4M \OP(\sqrt{\log(p)/ n_{min}}) + 4M \OP(s_0 \lambda) = \OP(s_0 \lambda),
	\end{align*}
	where the last inequality is a result of Lemma \ref{chap4:lemma:mom} and the fact that $\sup_{t \in [0, \tau]} \| \heta_{k}(t; \beta^0) - \heta_{k}(t; \hbeta) \|_{\infty} = \OP(\| \hbeta - \beta^0 \|_1) = \OP(s_0\lambda)$; see the proof of Lemma A4 in \citet{xia2021cox}. Since $\Sigma_{\beta^0} = \sum_{k=1}^K r_k \Sigma_{\beta^0, k}$,
	\begin{align*}
		a_{N2} \le & \displaystyle \left\| \sum_{k=1}^K \frac{n_k}{N} \left[  \frac{1}{n_k} \sum_{i=1}^{n_k} \int_{0}^{\tau} \{ X_{ki} - \eta_{k0}(t; \beta^0) \}^{\otimes 2} dN_{ki}(t) - \Sigma_{\beta^0, k} \right]  \right\|_{\infty} \\
		& + \displaystyle \left\|  \sum_{k=1}^K \left( \frac{n_k}{N} - r_k  \right) \Sigma_{\beta^0,k}  \right\|_{\infty} \\
		\le & \displaystyle  \sum_{k=1}^K \frac{n_k}{N} \left\| \frac{1}{n_k} \sum_{i=1}^{n_k} \int_{0}^{\tau} \{ X_{ki} - \eta_{k0}(t; \beta^0) \}^{\otimes 2} dN_{ki}(t) - \Sigma_{\beta^0, k}   \right\|_{\infty} + \displaystyle \left\|  \sum_{k=1}^K \left( \frac{n_k}{N} - r_k  \right) \Sigma_{\beta^0,k}  \right\|_{\infty}.
	\end{align*}
	The proof of Lemma A4 in \citet{xia2021cox} shows that, for $k=1, \cdots, K$,
	\[
	\displaystyle \left\| \frac{1}{n_k} \sum_{i=1}^{n_k} \int_{0}^{\tau} \{ X_{ki} - \eta_{k0}(t; \beta^0) \}^{\otimes 2} dN_{ki}(t) - \Sigma_{\beta^0, k}   \right\|_{\infty} = \OP(\sqrt{\log(p) / n_k})
	\] 
	by Hoeffding's concentration inequality.
	So $a_{N2} = \OP(\sqrt{\log(p)/ n_{min}}) + \mathcal{O}(\max_{k} | n_k/N - r_k|)$. Then, combining the bounds on $a_{N1}$ and $a_{N2}$, $\| \hSigma - \Sigma_{\beta^0} \|_{\infty} = \OP(s_0 \lambda + \max_k | n_k/N - r_k |)$.
	
	Finally, it is easy to see that
	\[
	\| \Theta_{\beta^0} \hSigma - I_p \|_{\infty} = \| \Theta_{\beta^0} (\hSigma - \Sigma_{\beta^0} ) \|_{\infty} \le \| \Theta_{\beta^0} \|_{1,1} \| \hSigma - \Sigma_{\beta^0} \|_{\infty},
	\]
	and $\| \Theta_{\beta^0} \hSigma - I_p \|_{\infty} = \OP ( \| \Theta_{\beta^0} \|_{1,1} \{ s_0 \lambda + \max_k | n_k/N - r_k | \})$.
\end{proof}

\begin{lemma} \label{chap4:lemma:diff_theta}
	Under the assumptions in Lemma \ref{chap4:lemma:const_sol}, if we further assume for some constant $\epsilon^{\prime} \in (0,1)$, $\limsup_{n_{min} \rightarrow \infty} p \gamma \le 1 - \epsilon^{\prime}$,  then we have $\|  \widehat{\Theta} - \Theta_{\beta^0} \|_{\infty} = \OP(\gamma \| \Theta_{\beta^0} \|_{1,1})$.
\end{lemma}

The proof of Lemma \ref{chap4:lemma:diff_theta} follows that of Lemma A5 in \citet{xia2021cox}, thus is omitted.

\begin{lemma} \label{chap4:lemma:max_score}
	Under Assumptions B.1--B.3, for each $t > 0$, 
	\[
	P (\| \dot{\ell}(\beta^0) \|_{\infty} > t) \le 2Kpe^{- C n_{min}t^2},
	\]
	where $C>0$ is an absolute constant.
\end{lemma}

\begin{proof}[\textbf{Proof of Lemma \ref{chap4:lemma:max_score}}]
	Since $\dot{\ell}(\beta^0) = \sum_{k=1}^K \frac{n_k}{N} \dot{\ell}^{(k)}(\beta^0)$, we have
	\begin{align*}
		P \left( \| \dot{\ell}(\beta^0) \|_{\infty} > t \right) & \le P \left( \sum_{k=1}^K \frac{n_k}{N} \| \dot{\ell}^{(k)}(\beta^0) \|_{\infty} > t \right) \\
		& \le \sum_{k=1}^K P \left( \| \dot{\ell}^{(k)}(\beta^0) \|_{\infty} > t \right) \\
		& \le \sum_{k=1}^K 2pe^{- C n_k t^2}.
	\end{align*}
	Note that $\| X_{ki} - \heta_k(t; \beta^0) \|_{\infty} \le 2M$ holds uniformly for all $k$ and $i$.
	Then the last inequality is a direct result of Lemma 3.3(ii) in \citet{huang2013oracle} when applied to each of the $K$ strata. 
\end{proof}

\section{Proofs of Main Theorems}

\begin{proof}[\textbf{Proof of Theorem 3.1}]
	Let $\dot{\ell}_j(\beta)$ be the $j$th element of the derivative $\dot{\ell}(\beta)$. By the mean value theorem, there exists $\widetilde{\beta}^{(j)}$ between $\hbeta$ and $\beta^0$ such that $\dot{\ell}_j(\hbeta) - \dot{\ell}_j(\beta^0) =  \left. \frac{\partial \dot{\ell}_j(\beta)}{\partial \beta^T} \right|_{\beta=\widetilde{\beta}^{(j)}} (\hbeta - \beta^0)$. Denote the $p \times p$ matrix $D = \left( \left. \frac{\partial \dot{\ell}_1(\beta)}{\partial \beta} \right|_{\beta=\widetilde{\beta}^{(1)}}, \cdots, \left. \frac{\partial \dot{\ell}_p(\beta)}{\partial \beta} \right|_{\beta=\widetilde{\beta}^{(p)}}   \right)^T$.
	By the definition of the de-biased estimator $\widehat{b}$, we may decompose $c^T (\widehat{b} - \beta^0) $ as
	\begin{align*} 
		c^T (\widehat{b} - \beta^0) & =  - c^T \Thetabeta \dot{\ell}(\beta^0) - c^T (\widehat{\Theta} - \Thetabeta) \dot{\ell}(\beta^0) \nonumber \\
		& \quad - c^T (\widehat{\Theta} \hSigma - I_p) (\hbeta - \beta^0) + c^T \hTheta (\hSigma -  D) (\hbeta - \beta^0) \\
		& = - c^T \Thetabeta \dot{\ell}(\beta^0) + (i) + (ii) + (iii),
	\end{align*} 
	where $(i) = - c^T (\widehat{\Theta} - \Thetabeta) \dot{\ell}(\beta^0), (ii) = - c^T (\widehat{\Theta} \hSigma - I_p) (\hbeta - \beta^0)$ and $(iii) = c^T \hTheta (\hSigma -  D) (\hbeta - \beta^0)$.
	
	We first show $\sqrt{N}(i) = \oP(1)$ and $\sqrt{N}(ii) = \oP(1)$.
	By Lemma \ref{chap4:lemma:diff_theta} and Lemma \ref{chap4:lemma:max_score},
	\begin{align*}
		| \sqrt{N}(i) | & \le \sqrt{N} \| c\|_1 \cdot \| \widehat{\Theta} - \Thetabeta \|_{\infty, \infty} \cdot \| \dot{\ell}(\beta^0) \|_{\infty} \\
		& \le \sqrt{N} a_{*} \OP( p \gamma \| \Thetabeta \|_{1,1}) \OP(\sqrt{\log(p)/n_{min}}) \\
		& = \OP( \| \Thetabeta \|_{1,1} p \gamma \sqrt{\log(p)}  ) \\
		& = \oP(1),
	\end{align*}
	where the last equality is a direct result of the assumption that $\| \Theta_{\beta^0} \|_{1,1}^2 \{ \max_k|n_k/N-r_k| + s_0\lambda \} p \sqrt{\log(p)} \to 0$ when $\lambda\asymp \sqrt{\log(p)/n_{min}}$.
	By Lemma \ref{chap4:lemma:lasso},
	\begin{align*}
		| \sqrt{N}(ii) | & \le \sqrt{N} \| c\|_1 \| (\widehat{\Theta} \widehat{\Sigma} - I_p) (\hbeta - \beta^0) \|_{\infty} \\
		& \le \sqrt{N} a_* \| \hTheta \widehat{\Sigma}  - I_p \|_{\infty} \| \hbeta - \beta^0 \|_1 \\
		& \le \sqrt{N} a_* \gamma \| \hbeta - \beta^0 \|_1 \\
		& = \OP(\sqrt{N} \gamma s_0 \lambda) \\
		& \le  \OP( \sqrt{N} \| \Thetabeta \|_{1,1} \{ \max_k |n_k/N - r_k| + s_0 \lambda \} p \sqrt{\log(p)/n_{min}} ) \\
		& = \oP(1).
	\end{align*}
	
	We then show that $\sqrt{N}(iii) = \oP(1)$.
	Note that $\hSigma - D = (\hSigma - \Sigma_{\beta^0}) + (\Sigma_{\beta^0} - \ddot{\ell}(\beta^0)) + (\ddot{\ell}(\beta^0) - D)$. By the proof of Lemma \ref{chap4:lemma:const_sol}, we see that with $\lambda \asymp \sqrt{\log(p) / n_{min}}$, $\| \hSigma - \Sigma_{\beta^0} \|_{\infty} = \OP(s_0 \lambda + \max_k | n_k/ N - r_k |)$. 
	Based on the proof of Theorem 1 in \citet{xia2021cox}, for each stratum, $\| \ddot{\ell}^{(k)}(\beta^0) - D^{(k)} \|_{\infty} = \OP(\sqrt{\log(p)/n_k})$, where
	$D^{(k)} = \left( \left. \frac{\partial \dot{\ell}^{(k)}_1(\beta)}{\partial \beta} \right|_{\beta=\widetilde{\beta}^{(1)}}, \cdots, \left. \frac{\partial \dot{\ell}^{(k)}_p(\beta)}{\partial \beta} \right|_{\beta=\widetilde{\beta}^{(p)}}   \right)^T$. Since the overall negative log partial likelihood $\displaystyle \ell(\beta) = \sum_{k=1}^K \frac{n_k}{N} \ell^{(k)}(\beta)$, and $D= \displaystyle \sum_{k=1}^K \frac{n_k}{N} D^{(k)}$, then $\| \ddot{\ell}(\beta^0) - D \|_{\infty} = \OP(\sqrt{\log(p)/n_{min}})$. Also, $\| \Sigma_{\beta^0, k} - \ddot{\ell}^{(k)}(\beta^0) \|_{\infty} = \OP(\sqrt{\log(p)/n_k})$ by the proof of Theorem 1 in \citet{xia2021cox}. Then 
	\begin{align*}
		\| \Sigma_{\beta^0} - \ddot{\ell}(\beta^0) \|_{\infty} & \le \left\| \sum_{k=1}^K r_k \Sigma_{\beta^0, k} - \sum_{k=1}^K \frac{n_k}{N} \Sigma_{\beta^0, k} \right\|_{\infty} + \left\| \sum_{k=1}^K \frac{n_k}{N} \Sigma_{\beta^0, k} - \sum_{k=1}^K \frac{n_k}{N} \ddot{\ell}^{(k)}(\beta^0) \right\|_{\infty} \\
		& \le K \max_k(| n_k / N - r_k | \| \Sigma_{\beta^0, k} \|_{\infty}) + K \OP(\sqrt{\log(p)/n_{min}}) \\
		& = \OP(\max_k |n_k/N - r_k| + \sqrt{\log(p)/n_{min}}).
	\end{align*}
	Therefore, for $\lambda \asymp \sqrt{\log(p)/n_{min}}$, $\| \hSigma - D \|_{\infty} = \OP(s_0 \lambda + \max_k |n_k/N - r_k|)$, and
	\begin{align*}
		| \sqrt{N} (iii) | & \le \sqrt{N} \| c \|_1 \| \hTheta \|_{\infty, \infty} \| \hSigma - D \|_{\infty} \| \hbeta - \beta^0 \|_1 \\
		& \le \OP\left(\sqrt{N} \| \Theta_{\beta^0} \|_{1,1} (s_0 \lambda +  \max_k | n_k/N - r_k |) \right) s_0 \lambda \\
		& \le \OP\left( \sqrt{N /n_{min}} \| \Theta_{\beta^0} \|_{1,1} (s_0 \lambda +  \max_k | n_k/N - r_k | ) p \sqrt{\log(p) } \right)\\
		& = \oP(1).
	\end{align*}
	
	Finally, for the variance,
	\begin{align*}
		| c^T ( \widehat{\Theta}  - \Thetabeta) c | & \le \| c \|_1^2 \| \widehat{\Theta}  - \Thetabeta \|_{\infty} \\
		& \le a_*^2 \OP(\gamma \| \Thetabeta \|_{1,1}) = \oP(1). 
	\end{align*}
	By Slutsky's theorem and Lemma \ref{chap4:lemma:lead}, $\sqrt{n} c^T (\widehat{b} - \beta^0) / (c^T \hTheta c)^{1/2} \overset{\mathcal{D}}{\rightarrow} N(0,1)$.
\end{proof}

\begin{proof}[\textbf{Sketch proof of Theorem 3.4}] Theorem 3.4 can be easily proved using Cram\'{e}r-Wold device. For any $\omega \in \mR^m$, since the dimension of $\omega$ is a fixed integer, we can invoke Theorem 3.1 by taking $c=J^T \omega$. Note that $\| J^T \omega \|_1 \le \| J^T \|_{1,1} \| \omega \|_1 = \| J \|_{\infty, \infty} \| \omega \|_1 = \mathcal{O}(1)$.
\end{proof}



\end{document}